\documentclass[11pt]{scrartcl}

\usepackage[latin1]{inputenc}
\usepackage[english]{babel}
\usepackage{amsmath,amsthm,amssymb,bbm,times,color}

\theoremstyle{plain}
\newtheorem{theorem}{Theorem}[section]
\newtheorem{lemma}[theorem]{Lemma}

\theoremstyle{definition}
\newtheorem*{definition}{Definition}

\theoremstyle{remark}

\numberwithin{equation}{section}

\newcommand{\R}{\mathbb{R}}
\newcommand{\C}{\mathbb{C}}
\newcommand{\Z}{\mathbb{Z}}

\newcommand{\id}{\boldsymbol{1}}

\DeclareMathOperator{\sgn}{sgn}
\DeclareMathOperator{\Imag}{Im}

\DeclareMathOperator{\tr}{tr}
\DeclareMathOperator{\pf}{pf}
\newcommand{\llangle}{\langle \!\langle}
\newcommand{\rrangle}{\rangle \!\rangle}
 \def\idty{{\mathchoice {\mathrm{1\mskip-4mu l}} {\mathrm{1\mskip-4mu l}} %
{\mathrm{1\mskip-4.5mu l}} {\mathrm{1\mskip-5mu l}}}}

\makeatletter
\let\Im\undefined

\DeclareMathOperator{\Im}{Im \,}

\makeatother

\title{Decay of Determinantal and Pfaffian Correlation Functionals in One-dimensional Lattices} 

\author{Robert Sims\\[10pt]
Department of Mathematics\\ University of Arizona\\
Tucson, AZ 85721, USA\\
Email: rsims@math.arizona.edu \\[10pt]
and \\[10pt]
Simone Warzel\\[10pt]
Zentrum Mathematik\\ TU M\"unchen\\
Boltzmannstr. 3, 85747 Garching, Germany\\
Email: warzel@ma.tum.de}

\begin{document}

\maketitle

\begin{abstract} 
We establish bounds on the decay of time-dependent multipoint correlation functionals of one-dimensional quasi-free fermions in terms of the decay properties of their two-point function. At a technical level, this is done with the help of bounds on certain bordered determinants and pfaffians.
These bounds, which we prove, go beyond the well-known Hadamard estimates. 
Our main application of these results is a proof of strong (exponential) dynamical localization of spin-correlation functions in disordered $XY$-spin chains. 
\end{abstract} 

\section{Introduction and main results}

\subsection{Quasi-free Fermions on the lattice}
Systems of (spinless) fermions are described on the fermionic Fock space $ \mathcal{F}(\mathcal{H}) $ over a separable single-particle Hilbert space $ \mathcal{H} $.  Fermionic annihilation and creation operators  $ c(f) $ and $ c^*(f) $ associated with $ f \in \mathcal{H} $ act on $ \mathcal{F}(\mathcal{H}) $ and satisfy the canonical anticommutation relations (CAR), i.e., for any  $ f, g \in \mathcal{H} $ 
\begin{equation}
 \left\{ c(f),  c^*(g)  \right\} =   c(f)  c^*(g) +    c^*(g) c(f) = \langle f, g \rangle \idty
\end{equation}
and all other anticommutators vanish. By $\mathcal{A}( \mathcal{H})$, we denote the $C^*$-algebra generated by the identity, $\idty$, and the operators
$c(f)$ and $c^*(g)$ for all $f,g \in \mathcal{H}$; this is the CAR algebra associated to $\mathcal{H}$. 
A state $\omega$ on $ \mathcal{A}(\mathcal{H}) $ is any normalized, non-negative linear functional.  The state $\omega$, or more generally any linear functional, 
is called quasi-free (or: quasi-gaussian, cf.~\cite{A,BLS94})  if all correlation functions are computed from Wick's theorem.
More precisely, $ \omega $  is quasi-free if for any $ n \in \mathbb{N} $ and any collection $ f_j , g_j \in  \mathcal{H} $ giving rise to operators
\begin{equation}
 C_j := c(f_j) + c^*(g_j) \, , \quad  j = 1 \dots , 2n \, , 
 \end{equation}
we have $ \omega( C_1 \cdots  C_{2n-1} ) = 0 $ and 
\begin{equation}\label{def:quasifree}
\omega( C_1 \cdots  C_{2n} ) = \sum_\pi {}' \sgn \pi  \; \omega\left(C_{\pi(1)} C_{\pi(2)}\right) \cdots  \omega\left(C_{\pi(2n-1)} C_{\pi(2n)}\right) \quad \big[  =: \pf \mathcal{C} \, \big] \, .
\end{equation}
The (primed) sum is over all permutations $ \pi $ which satisfy $ \pi(1) < \pi(3) < \dots < \pi(2n-1) $ and $ \pi(2j-1) < \pi(2j) $ for all $ 1 \leq j \leq n $. The right side is known as the pfaffian of the triangular array 
\begin{equation}\label{eq:triangular}
 \mathcal{C} := \left(  \omega(C_j C_k) \right)_{1\leq j < k \leq 2n }  
\end{equation}
which equivalently  may be identified with a skew-symmetric matrix,  
$ \mathcal{C}^T = - \mathcal{C} $.

In general, fermionic quasi-free states are characterized by 
$
 \omega(C_j C_k)  = \langle \begin{pmatrix} f_k \\ \overline{g_k} \end{pmatrix} , \Gamma \begin{pmatrix} g_j \\ \overline{f_j} \end{pmatrix} \rangle $, in terms of a one-particle density matrix $  \Gamma  $ on $ \mathcal{H} \oplus \mathcal{H} $. The latter
satisfies $ 0 \leq \Gamma \leq 1 $ and is of the from $ \Gamma =  \begin{pmatrix}  \varrho & \alpha \\ \alpha^* & 1- \overline{\varrho} \end{pmatrix} $ where $  \langle \overline{g_k} , (1- \overline{\varrho} )  \overline{f_j}  \rangle := \langle f_j, (1-\varrho) \ g_k \rangle $, cf.~\cite{BLS94}.
For a simple ubiquitous subclass of quasi-free states $ \alpha = 0 $ and hence they are uniquely characterized by a one-particle density operator $ 0 \leq \varrho \leq 1 $ on $ \mathcal{H} $,
\begin{equation}\label{eq:normalstate}
\omega_\varrho\left( C_j C_k \right) = \big\langle f_k, \varrho \ g_j \rangle +   \langle f_j, (1-\varrho) \ g_k \big\rangle \,  .
\end{equation} 
Examples of such one-particle density operators which will play some role in this paper include:
\begin{enumerate}
\item
the thermal equilibrium state  corresponding to a single-particle Hamiltonian $ H $ on $ \mathcal{H} $ with inverse temperature $ \beta > 0 $ and chemical potential $ \mu \in \mathbb{R} $:
\begin{equation}\label{eq:thermalstate}
\varrho = \left( 1 + e^{\beta(H-\mu) } \right)^{-1} \, . 
\end{equation}
In the limit $ \beta \to \infty $, this turns into the fermionic ground state, $ \varrho = P_{(-\infty,\mu)}(H) $, i.e., the spectral projection of $ H $ corresponding to energies strictly below $ \mu$. 
\item those of the form 
\begin{equation}\label{eq:eigenstates}
\varrho_\alpha = \sum_{j : \, \alpha_j = 1 } | \phi_j \rangle \langle \phi_j | 
\end{equation}
with $ \alpha \in \{ 0,1 \}^{\mathbb{N} } $ and $ ( \phi_j )_{j \in \mathbb{N}} $ an orthonormal basis of eigenvectors 
of a single-particle Hamiltonian $ H $ on $ \mathcal{H} $. For $ \alpha $ with the property that $\alpha_j=1$ whenever $ \lambda_j < \mu $ and $\alpha_j=0$ otherwise, this state coincides with the ground state of the non-interacting Fermi system with Fermi energy $ \mu $. Other choices of $ \alpha $ correspond to excited states.
\end{enumerate}
We will be interested in lattice fermions for which $ \mathcal{H} = \ell^2(\mathbb{Z}^d) $ -- and primarily with $ d = 1 $. In the lattice case, for any state  $ \omega_\varrho $ 
of the form~\eqref{eq:normalstate}, the canonical orthonormal basis of 
vectors $ \delta_{\xi} \in \ell^2(\mathbb{Z}^d) $ localized at the lattice sites $ \xi \in \mathbb{Z}^d $ gives rise to time-dependent multipoint correlation functions with determinantal structure, i.e., 
for $ x_1, \dots , x_n , y_1 , \dots , y_n \in \mathbb{Z}^d $ and $s_1 , \dots , s_n , t_1, \dots , t_n \in \mathbb{R} $:
\begin{multline}\label{eq:multipoint-static}
 \omega_\varrho\left( c^*(e^{it_nH} \delta_{y_n}) \cdots c^*(e^{it_1H} \delta_{y_1}) \, c(e^{is_1H}\delta_{x_1}) \cdots c(e^{is_nH} \delta_{x_n})\right) \\
 = \det\left( \langle (e^{is_jH}  \delta_{x_j} , \varrho \, e^{it_kH} \delta_{y_k} \rangle \right)_{1\leq j , k \leq n } \, . 
\end{multline}
Here for a given single-particle Hamiltonian $ H $ on $ \ell^2(\mathbb{Z}^d) $, a simple (free) dynamics is implemented on the corresponding CAR algebra $ \mathcal{A}_d := \mathcal{A}(\ell^2(\mathbb{Z}^d)) $ through
\begin{equation}\label{eq:freedyn}
\tau_t ( c(f) ) := c(e^{-itH}f) \, , \quad t \in \mathbb{R} \, ;
\end{equation}
more general dynamics will be considered later.
(The derivation of~\eqref{eq:multipoint-static} from~\eqref{eq:normalstate}  proceeds by setting $ f_1=\dots = f_n = g_{n+1} = \dots = g_{2n} = 0 $ and $ g_j = e^{it_{n-j+1}H} \delta_{y_{n-j+1}} $, $ f_{n+j} = e^{is_jH} \delta_{x_j} $ for $ j = 1, \dots n $ and using the fact that the arising triangular array~\eqref{eq:triangular} has a block structure in which case the pfaffian reduces to a determinant.)\\

\subsubsection{Determinant bound}
In its simplest form, the basic question of this paper can already be formulated in this setting. Suppose the two-point function decays, i.e., for some $ C, \mu \in (0,\infty) $ and all $ x_j , y_k \in \mathbb{Z}^d $ 
\begin{equation}\label{eq:expdecay1}
\sup_{s,t \in \R}  \left| \langle e^{isH} \delta_{x_j} , \varrho \,  e^{itH} \delta_{y_j} \rangle \right| \ \leq \ C \, e^{-\mu K(|x_j-y_k|) } \, , 
\end{equation}
with some monotone increasing $ K: [0,\infty) \to [0,\infty) $, 
what can one say about the decay of  multipoint correlation functions such as~\eqref{eq:multipoint-static}? 
Before describing our answer to this question, let us first clarify some points.

Exponential decay of the two-point function would correspond to the choice $ K(\tau)  = \tau  $. Slower decay rates such as algebraic decay may be accommodated by another choice such as
$ K(\tau) = \ln(1+ \tau) $. Our basic assumption throughout this paper is:
\begin{description}
\item[Assumption:] \quad  $ K: [0,\infty) \to [0,\infty) $ is monotone increasing and 
there exists some $ \mu_0 \in (0,\infty) $ such that
\begin{equation}
I(\mu_0) := \sum_{\ell=0}^\infty (1+\ell) \, e^{-\mu_0 K(\ell)} < \infty \, . 
\end{equation}
\end{description}

Since the  multipoint correlation functions depend on configurations
$ x = ( x_1, \dots , x_n) \in \mathbb{Z}^{d n} $ and $ y = ( y_1, \dots , y_n) \in \mathbb{Z}^{d n} $, we need to specify our notion of distance for configurations. Since we focus here on one-dimensional systems, $ d = 1 $, we may restrict without loss of generality to fermionic configurations which are naturally ordered: 
\begin{equation}\label{eq:order1}
x_1 < x_2 < \dots < x_n \, , \qquad y_1 < y_2 < \dots < y_n \, . 
\end{equation}
Note that there is only one particle per site, since we are dealing with (spinless) fermions.
In this situation, the notion of distance we will adopt is:
\begin{equation}\label{eq:distance}
	D(x,y) := \max_{j \in \{ 1, \dots, n\} }|x_j-y_j|\, . 
\end{equation}
Some remarks are in order:
\begin{enumerate}
\item
Since we deal with configurations of indistinguishable particles, one might wonder whether this distance is invariant under relabeling of particles. This immediately follows from the fact that
$ D(x,y) = \min_{\pi}  \max_{j \in \{ 1, \dots, n\} } |x_j-y_{\pi(j)}| $ where the minimum is over all permutations $ \pi $ of $ n $ elements. 
\item 
There are other notions of distance then \eqref{eq:distance} which are natural for indistinguishable particles. For example,  $ D_1(x,y) :=  \min_{\pi} \sum_{j=1}^n  |x_j-y_{\pi(j)}| $ is an  option which in our one-dimensional situation turns out to be $  D_1(x,y) = \sum_{j} |x_j-y_j| $ in case of ordered fermionic configurations~\eqref{eq:order1}. 
\end{enumerate}
One result of this paper is the following:
\begin{theorem}[Determinant bound]\label{thm:det}
Let $I \subseteq \mathbb{R}$ and $ \rho(s,t)  $ be a family of uniformly bounded operators on $ \ell^2(\mathbb{Z}) $ with $ \| \rho(s,t)  \| \leq 1 $ for all $ s, t \in I $. If there is some $ C \in(0,\infty)$, $\mu \in (\mu_0,\infty) $ such that for all $ x,y \in \mathbb{Z} $:
\begin{equation}\label{eq:expdecay}
\sup_{s,t \in I} \left| \langle \delta_{x} , \rho(s,t)    \delta_{y} \rangle \right| \ \leq \ C \, e^{-\mu K(|x-y|) } \, ,
\end{equation}
then for any $ n \in \mathbb{N} $ and any pair of fermionic configurations $ x= ( x_1, \dots , x_n) , y = ( y_1, \dots , y_n) \in \mathbb{Z}^n $:
\begin{multline}\label{eq:det}
\sup_{s,t \in I^n} \left| \det\left( \langle \delta_{x_j} , \rho(s_j, t_k)  \delta_{y_k} \rangle \right)_{1\leq j , k \leq n }   \right| \\ 
\leq 8 \max\{ C I(\mu_0) , \sqrt{CI(\mu_0)} \}  \,  \exp\left(-\tfrac{\mu-\mu_0}{2} \, K\left(\tfrac{D(x,y)}{2}\right)\right) \, . 
\end{multline}
\end{theorem}
Some remarks apply:
\begin{enumerate}
\item  The theorem does not require $ \rho(s,t) $ to be a reduced density operator -- only the norm bound $ \| \rho(s,t) \| \leq 1 $ is essential. In particular, the choice $ \rho (s,t) = e^{-isH} \varrho \, e^{itH} $ with some reduced density operator $ 0 \leq \varrho \leq 1 $ is admissible, cf.~\eqref{eq:multipoint-static}. 

\item Exponential decay of correlations, i.e., \eqref{eq:expdecay} with $ K(\tau) = \tau $,  is known to occur for thermal states of single-particle systems. More precisely, in case $ \rho = \rho(H) $ with a self-adjoint operator $ H $ on $ \ell^2(\mathbb{Z}^d) $ and any function $ \rho: \mathbb{R} \to \mathbb{C} $, which has an analytic extension to the strip $|  \Imag z | \leq \eta $ on which it is bounded by $ \| \rho \|_\infty $, obeys for all $ x_j,y_k \in \mathbb{Z}^d $:
\begin{equation}\label{eq:CombesThomas}
\left| \langle \delta_{x_j} , \rho(H)  \, \delta_{y_j} \rangle \right| \ \leq \ 18 \sqrt{2} \, \| \varrho \|_\infty  \, e^{-\mu |x_j - y_k | } 
\end{equation}
for any $ \mu > 0 $ such that $ b(\mu) := \sup_{\xi' \in \mathbb{Z}^d}   \sum_{\xi \in \mathbb{Z}^d}  | \langle \delta_{\xi'} , H \delta_\xi \rangle | ( e^{\mu |\xi-\xi'|} - 1 ) < \eta / 2 $, cf.~\cite[Theorem~3]{AG98}. 

In particular, this applies with $ \eta < \pi/\beta $ to the Fermi distribution function, $ \rho(H) = (1+e^{\beta (H-\mu)} )^{-1}  $. 
\end{enumerate}

\subsubsection{Disordered case}
Part of the motivation for Theorem~\ref{thm:det} stems from the analysis of disordered systems. Here the single-particle Hamiltonian $ H $ (as well as the one-particle density matrix $ \varrho $) is a weakly measurable map from some probability space $ (\Omega, \Sigma, \mathbb{P} ) $ into the space of self-adjoint operators on $ \ell^2(\Z^d) $. The most prominent examples of such $ H $ are Anderson-type opeartors, i.e., discrete Schr\"odinger operators of the form $ H = - \Delta + V $ where the multiplication operator $ V $ is given by independent and identically distributed (iid) random variables $ (\nu_\xi ) $ associated to $ \xi \in \mathbb{Z}^d $. The spectral theory of such operators is quite well studied -- the main feature being the existence of a localized phase~\cite{CL,PF,AW}. 
For convenience of the reader, let us  summarize some facts which are important in the following:
\begin{enumerate}
\item
A dynamical characterization of localization involves the eigenfunction correlator which is defined as the total variation measure associated with $ x , y \in \mathbb{Z}^d $: 
\begin{equation}
Q(x,y;I) := \sup_{\substack{f \in L^\infty(\mathbb{R}) \\ \| f \|_\infty \leq 1 }}  \left| \langle \delta_{x} , f(H)  P_I(H) \, \delta_{y} \rangle \right| 
\end{equation}
where $ P_I(H) $ denotes the spectral projection on $ I \subseteq \mathbb{R} $. Strong exponential dynamical localization in $ I $ then refers to the bound:
\begin{equation}\label{eq:expdecaysl}
\mathbb{E}\left[ Q(x,y;I)   \right] \ \leq \ C \, e^{-\mu |x - y | } \, 
\end{equation}
for some $ C, \mu \in (0,\infty) $ and all $ x,y \in \mathbb{Z}^d $. It implies
that the spectrum of $ H $ is almost-surely pure point. In case the latter is simple given by $ \lambda_1 < \lambda_2 < \dots $, the eigenfunction correlator is then $ Q(x,y;I) = \sum_{\lambda_j\in \sigma(H) \cap I } |\phi_j(x)| |\phi_j(y)| $, given in terms of the normalized eigenbasis $ \{ \phi_j \} $ of $ H $. 
\item
Under some reasonable assumptions on the random operator $ H $, strong exponential dynamical localization~\eqref{eq:expdecaysl} is known to occur in case $ d = 1 $ throughout the spectrum, i.e. with $ I = \mathbb{R} $; for details, see~\cite{CL,AW} and references therein. 
\end{enumerate}
The assumption of the following theorem hence applies to 
such random operators $ H $ on $\ell^2(\mathbb{Z}) $ and all (time-evolved) one-particle operators $  \rho(s,t) = \rho(H)  \left( e^{i (t-s) H} \right) $ which are bounded functions of the Hamiltonian such as, for example, thermal-states~\eqref{eq:thermalstate} up to $ \beta = \infty $ or
eigenstates~\eqref{eq:eigenstates} related to $ H $. 
\begin{theorem}[Strong dynamical localization I]\label{thm:detrandom}
Let $I \subseteq \mathbb{R}$ and consider a family  of random operators $\rho(s,t)  $ on $ \ell^2(\mathbb{Z}) $ with $ \| \rho (s,t)\| \leq 1 $ for all $ s, t \in I $ which exhibit localization in the sense that for some $ C, \mu \in (0,\infty) $ and for all $ x, y \in \mathbb{Z} $:
\begin{equation}\label{ass:dis}
\mathbb{E}\left[ \sup_{s,t\in I} \left| \langle \delta_{x} ,  \rho(s,t) \, \delta_{y} \rangle \right| \right] \ \leq \ C \, e^{-\mu |x - y | } \, .
\end{equation}
Then for any $ n \in \mathbb{N} $ and any pair of configurations $ x= ( x_1, \dots , x_n) , y = ( y_1, \dots , y_n) \in \mathbb{Z}^n $:
\begin{equation}\label{eq:detrandom}
\mathbb{E}\left[ \sup_{s,t\in I^n} \left| \det\left( \langle \delta_{x_j} , \rho(s_j,t_k) \, \delta_{y_k} \rangle \right)_{1\leq j , k \leq n }   \right| \right] \leq \frac{8 \max\{ C, \sqrt{C}\}}{(1 - e^{-\mu})^2 }\,  \exp\left(-\tfrac{\mu}{4} \, D(x,y)\right) \, . 
\end{equation}
\end{theorem}
Several remarks apply:
\begin{enumerate}
\item The case $ \rho(s,t) = e^{-itH} $ in Theorem~\ref{thm:detrandom} includes a statement on the determinant of the time evolution operator $ e^{-itH } $ projected to a pair of configurations $ x , y \in \mathbb{Z}^n $. In \cite{BK} this quantity arises in the analysis of an error-correcting code for a one-dimensional chain of Majorana fermions. More precisely, the random one-particle Hamiltonian $ H $ on $ \ell^2(\{1,\dots, N\}) $ is dubbed in~\cite{BK} \emph{multipoint dynamical localized} if there are constants $ C , \mu \in (0,\infty) $ such that for all $ n \leq N  $ sufficiently large:
\begin{equation}\label{eq:mpdl}
\sup_{t\in \mathbb{R}} \,  \mathbb{E}\left[  \left| \det\left( \langle \delta_{x_j} , e^{itH} \delta_{y_k} \rangle \right)_{1\leq j , k \leq n } \right| \right] \ \leq \ C^n e^{-\mu N} \, , 
\end{equation}
for all configurations $ x , y \in \mathbb{Z}^n $ with $ D_1(x,y) \geq N/8 $.

Eq,~\eqref{eq:detrandom} is weaker in case $ 1 \ll n \ll N $, since 
$ D_1(x,y) \leq n \, D(x,y)  $. It is an interesting open question whether~\eqref{eq:mpdl} holds in the regime of strong-dynamical one-particle localization~\eqref{eq:expdecaysl}.     
\item The exponential decay in (\ref{ass:dis}), and then subsequently in (\ref{eq:detrandom}), can also be replaced by a slower or faster decay (captured by $ K $) as in Theorem~\ref{thm:det}. 
It is an interesting open question whether the above Theorems~\ref{thm:det} and \ref{thm:detrandom} can be generalized to higher dimensions. 
\end{enumerate}

\subsection{Majorana Fermions}\label{ch:Majorana}

For lattice fermions one may associate to each site $ x \in \mathbb{Z}^d $ a pair of Majorana fermions:
\begin{align}\label{def:majorana}
 a_x^+  := c^*(\delta_x) +  c(\delta_x)  \, , \quad 
 a_x^-  :=  i \left( c^*(\delta_x)  - c(\delta_x) \right) \, ,
\end{align}
These operators are self-adjoint,  $  (a_x^\# )^* =  a_x^\# $, and  satisfy $ (a_x^\# )^2 = 1 $ for both $ \# = \pm  $. They obey the anticommutation relations
\begin{equation}\label{MajoranaCAR}
 \left\{ a_x^\# , a_y^\flat \right\} = 2 \delta_{x,y} \delta_{\#,\flat} \, \idty . 
 \end{equation}
The free dynamics~\eqref{eq:freedyn} generated by a single-particle Hamiltonian $ H $ carries over to the Majorana fermions
\begin{equation}
a_x^\#(t) := \tau_t\left( a_x^\# \right) \, . 
\end{equation}
Given such a dynamics and a quasi-free state $ \omega $  the dynamical multipoint correlation functions 
are of the form
\begin{equation}\label{eq:quasifree}
\omega\left( a_{x_1}^{\#_1}(t_1) \dots a_{x_{2n}}^{\#_{2n}}(t_{2n}) \right) = \pf\left( \omega\left(   a_{x_j}^{\#_j}(t_j)  a_{x_k}^{\#_k}(t_k) \right)\right)_{1\leq j < k \leq 2n }  \, . 
\end{equation}
In the subsequent theorem, we envoke the following definition.
\begin{definition}
A pair $ (\omega,\tau) $ of a functional $ \omega $ and automorphisms $ \tau = \{ \tau_t \}_{t\in \mathbb{R}} $ on the CAR algebra $\mathcal{A}_d:= \mathcal{A}(\ell^2(\mathbb{Z}^d)) $ is called quasi-free, if the Wick relation~\eqref{eq:quasifree} holds for all $ n \in \mathbb{N} $ at all  $  (x_1, \#_1), \dots , (x_{2n},\#_{2n}) \in \mathbb{Z}^d \times \{\pm \} $, and all times $ (t_1,\dots , t_{2n}) \in \mathbb{R}^{2n} $.  
\end{definition}
As before, our main concern will be the decay rate of such multipoint correlation functions, given information about the decay of the two-point function. Since the multipoint correlation function involves
a collection $ x := (x_1,\dots , x_{2n} ) \in \mathbb{Z}^{2nd} $ of points, we again first need to quantify the relevant notion of distance for this collection. 
In the one-dimensional situation, $ d = 1 $, the points may be ordered without loss of generality
\begin{equation}\label{eq:ordermaj} 
 x_1 \leq x_2 \leq \dots \leq x_{2n-1} \leq x_{2n} \, . 
\end{equation}
Note that in contrast to~\eqref{eq:order1} these points are not necessarily distinct since $ x_j \in \mathbb{Z} $ may carry two Majoranas: one with $ \# = + $ and one with $ \# = - $. We will call $  (x_1, \#_1), \dots , (x_{2n},\#_{2n}) \in \mathbb{Z} \times \{\pm \} $ a Majorana configuration if these tupels are distinct for all $ j \in \{ 1,\dots , 2n\} $. A natural notion of distance for such an ordered Majorana configuration is
\begin{equation}
r(x) := \max_{j\in \{ 1, \dots , n \} } | x_{2j} - x_{2j-1} | \, . 
\end{equation}
In this context our first main result  then reads as follows:
\begin{theorem}[Pfaffian bound]\label{thm:pf}
Let $ (\omega, \tau) $ be a quasi-free pair on the CAR algebra $ \mathcal{A}_1  $ and assume that $ \omega $ is a bounded functional, i.e., there is
some $M_0 \in (0, \infty)$ such that $ |\omega(A) | \leq M_0 \| A \| $ for all $ A \in \mathcal{A}_1$.
Let $ I \subset \mathbb{R} $ and suppose there is some $ C \in(0,\infty)$, $\mu \in (\mu_0,\infty) $ such that for all $ x,y \in \mathbb{Z} $:
\begin{equation}\label{eq:expdecay33}
\max_{\#,\flat \in \{\pm \} } \sup_{s,t \in I} \left|\omega\left(   a_{x}^{\#}(t)  a_{y}^{\flat}(s) \right) \right| \ \leq \ C \, e^{-\mu K(|x-y|) } \, .
\end{equation}
Then there is some $ C'=C'(\mu_0) $ such that for any $ n \in \mathbb{N} $, and any Majorana configuration $ (x_1, \#_1), \dots , (x_{2n},\#_{2n}) \in \mathbb{Z} \times \{\pm \}  $:
\begin{equation}\label{eq:pf}
\sup_{t \in I^{2n}} \left| \pf\left( \omega\left(   a_{x_j}^{\#_j}(t_j)  a_{x_k}^{\#_k}(t_k) \right)\right)_{1\leq j < k \leq 2n }    \right| \\ \leq  M_0 \, C'(\mu_0) \,   e^{-(\mu -\mu_0) K(r(x)) } \, . 
\end{equation}
\end{theorem}

Let us stress that we do not assume here that $ \omega $ is a state: it neither needs to be non-negative nor normalized. Only its Gaussian nature and boundedness are essential.

Similar to before, this bounds carries over to the random case.

\begin{theorem}[Strong dynamical localization II]\label{thm:majodis}
Let $ (\omega, \tau) $ be a random quasi-free pair on the CAR algebra $\mathcal{A}_1  $ and assume that $ \omega $ is a bounded functional, i.e., 
there is some random $ M_0 \in (0,\infty) $ such that $ |\omega(A) | \leq M_0 \| A \| $ for all $ A \in \mathcal{A}_1$. 
Let $ I \subset \mathbb{R} $ and suppose there is some (non-random) $ C, \mu  \in(0,\infty)$ such that for all $ x,y \in \mathbb{Z} $:
\begin{equation}\label{ass:majdis}
\max_{\#,\flat \in \{\pm \} } \mathbb{E} \left[\sup_{s,t \in I} \left|\omega\left(   a_{x}^{\#}(t)  a_{y}^{\flat}(s) \right) \right| \right] \ \leq \ C \, e^{-\mu |x-y| } 
\end{equation}
Then there is some $ C' = C'(\mu) \in (0,\infty) $ such that  for any $ n \in \mathbb{N} $, and any Majorana configuration $ (x_1, \#_1), \dots , (x_{2n},\#_{2n}) \in \mathbb{Z} \times \{\pm \}  $:
\begin{equation}\label{eq:pf233}
\mathbb{E}\left[\sup_{t \in I^{2n}} \frac{1}{M_0} \left| \pf\left( \omega\left(   a_{x_j}^{\#_j}(t_j)  a_{x_k}^{\#_k}(t_k) \right)\right)_{1\leq j < k \leq 2n }    \right| \right]  \leq   C'(\mu) \,  e^{- \mu  r(x) /3   } \, . 
\end{equation}
\end{theorem}
Let us conclude with two remarks:
\begin{enumerate}
\item It is straightforward to see from the subsequent proof that we may extend Theorem~\ref{thm:majodis} to the case in which $ \omega $ depends on an additional parameter $\alpha $. If one assumes uniform exponential decay in the sense that~\eqref{ass:majdis} holds with an additional supremum over $ \alpha $ inside the expectation, then~\eqref{eq:pf233} holds with an additional supremum inside the expectation. 
\item Again, the above theorem has a straightforward generalization to the case that the two-point function decays at a rate given by  $ K $.
\end{enumerate}

\section{Time-dependent correlations in random $XY$ spin chains}
Our main application of Theorem~\ref{thm:majodis} concerns the correlation functions of a (random) spin-$\tfrac{1}{2}$  chain.
More precisely, we consider an anisotropic spin chain of length $N \in \mathbb{N} $ with the Hamiltonian
\begin{equation} \label{xychain}
S_N = - \sum_{\xi=1}^{N-1} \mu_\xi [ (1+\gamma_\xi) \sigma_\xi^1 \sigma_{\xi+1}^1 + (1-\gamma_\xi) \sigma_\xi^2 \sigma_{\xi+1}^2] - \sum_{\xi=1}^N \nu_\xi \sigma_\xi^3 
\end{equation}
which  acts on the Hilbert space $\mathcal{H}_N = \bigotimes_{\xi =1}^N \mathbb{C}^2$. 
The real-valued sequences$\{ \mu_\xi \}$, $\{ \gamma_\xi \}$, and $\{ \nu_\xi \}$ are the parameters 
of the model which can be physically interpreted as an interaction strength, the anisotropy, and an 
external magnetic field in $ 3 $-direction, respectively, and 
\begin{equation}
\sigma^1 = \left( \begin{array}{cc} 0 & 1 \\ 1 & 0 \end{array} \right), \quad \sigma^2 = \left( \begin{array}{cc} 0 & -i \\ i & 0 \end{array} \right), \quad \mbox{and} \quad 
\sigma^3 = \left( \begin{array}{cc} 1 & 0 \\ 0 & -1 \end{array} \right)
\end{equation}
denote the Pauli matrices. By the subscripts $\xi \in \{1,\dots,N\}$, we embed these matrices into $\mathcal{B}(\mathcal{H}_N)$, i.e., $ \sigma_\xi^w = \idty \otimes \cdots \otimes \idty \otimes \sigma^w \otimes \idty \otimes \cdots \otimes \idty $ for any
$w \in \{1,2 ,3 \}$
with $\sigma^w$ appearing in the $\xi$th factor. 

The dynamics generated by the Hamiltonian $S_N$ is the one-parameter group of
automorphisms on $\mathcal{B}( \mathcal{H}_N)$ given by
\begin{equation}
\tau_t^N(A) = e^{itS_N} A e^{-itS_N} \quad \mbox{for all } A \in \mathcal{B}( \mathcal{H}_N) \mbox{ and } t \in \mathbb{R} \, .
\end{equation}
We are interested in dynamic correlations between general single-site observables. More concretely, for any
$1 \leq \xi \leq N$ denote by $\mathcal{A}_{\{ \xi\}}$ the set of observables with support $\{ \xi \}$.
With $1 \leq \xi < \eta \leq N$ fixed, let $A \in \mathcal{A}_{\{ \xi \}}$ and $B \in \mathcal{A}_{\{ \eta \}}$. We consider
\begin{equation}\label{eq:corr0}
\langle \tau_t(A) B \rangle - \langle \tau_t(A) \rangle \langle B \rangle  \, , \quad \mbox{with}\quad 
 \langle \cdot \rangle := \tr \left( \rho(S_N) \; (\cdot) \right) \, .
\end{equation} 
The states $\langle \cdot \rangle$ are described in terms of their density matrices $ \rho(S_N) \geq 0 $. 
We will mainly consider either eigenstates or thermal states associated to $ S_N $, i.e., 
\begin{align}\label{def:state}
 \rho(S_N)  = \begin{cases}
 	| \Psi_\alpha \rangle   \langle \Psi_\alpha | \, , \quad & \mbox{eigenstate of $ S_N $ with label $ \alpha $,}\\
	  e^{-\beta S_N} / \tr e^{-\beta S_N }  \, , & \mbox{thermal state with inverse temperature $ \beta $.}  
	  \end{cases}
\end{align}
In order to distinguish the two cases, we will sometimes include a subscript $ \alpha $ (in case of an eigenstate to be described below) or $ \beta $ (in case of a thermal state). 
Note that since these are expectations in a state whose density matrix commutes with $ S_N $, it is clear that they are time invariant, i.e. $\langle \tau_t(A) \rangle = \langle A \rangle $. 
To calculate the correlations~\eqref{eq:corr0}, we first expand the single-site observables in terms of a basis. Any $A \in \mathcal{A}_{\{ \xi \}}$ can be written as:
$
A = a_0 \idty + a_1 \sigma_\xi^1 + a_2 \sigma_\xi^2  + a_3 \sigma_\xi^3 = \sum_{w=0}^3 a_w \sigma^w_\xi
$
and we have set $\sigma_\xi^0 = \idty$ for convenience. As a result,
\begin{equation}\label{eq:corrfct}
\langle \tau_t(A) B \rangle - \langle A \rangle \langle B \rangle = 
\sum_{w, w'=1}^3 a_w b_{w'} \left( \langle \tau_t(\sigma_\xi^w) \sigma_\eta^{w'} \rangle - \langle \sigma_\xi^w \rangle \langle \sigma_\eta^{w'} \rangle \right) \, .
\end{equation}
In order to estimate these correlation functions, we relate them to correlations of free Majorana fermions using the well known Jordan-Wigner transformation~\cite{JW,LM}.

\subsection{Jordan-Wigner transformation in terms of Majorana Fermions}\label{ss:JW}
The operators
\begin{align} \label{cdef}
& a_1^+ = \sigma^1_1 \qquad  \mbox{and} \quad a_\xi^+ = \sigma_{1}^3 \cdots \sigma_{\xi-1}^3 \sigma^1_\xi \quad \mbox{for all } 2 \leq \xi \leq N \, , \notag  \\
& a_1^- = - \sigma_1^2 \quad \mbox{and} \quad a_\xi^- = - \sigma_{1}^3 \cdots \sigma_{\xi-1}^3 \sigma^2_\xi \quad \mbox{for all } 2 \leq \xi \leq N \, ,
\end{align}
are self-adjoint, $ (a_\xi^\# )^* = a_\xi^\# $, and satisfy $ (a_\xi^\# )^2 = 1 $ as well as the anti-commutation rules~\eqref{MajoranaCAR} for Majorana fermions.
A short calculation also shows that 
$
i a_\xi^+ a_\xi^- = \sigma_\xi^3 $
and that the Hamiltonian coincides with the following quadratic form 
\begin{equation} \label{eq:Crep}
S_N =   \frac{1}{2} \,  \mathcal{A}^T H_N \mathcal{A}
\end{equation}
in terms of the vector $ 
\mathcal{A}= (a_1^+, a_1^-, \ldots, a_N^+, a_N^-)^T $. 
The $2N \times 2N$ coefficient matrix $H_N$ is  self-adjoint and of Jacobi block-form (with blocks composed of Pauli matrices):
\begin{equation}\label{blockmatrix}
H_N =\begin{pmatrix}
            \nu_1 \sigma^2 & - \mu_1S(\gamma_1) & &  \\
          - \mu_1 S(\gamma_1)^* &    \nu_2 \sigma^2 &  \ddots  &  \\
           & \ddots & \ddots &   -\mu_{N-1} S(\gamma_{N-1}) \\
           &  &  -\mu_{N-1} S(\gamma_{N-1} )^*  & \nu_N \sigma^2
          \end{pmatrix} \, 
\end{equation}
where $S(\gamma) := \sigma^2 + i \gamma \sigma^1$. The operator $ H_N  $ acting on $ \ell^2(\{1,\dots, N\} ; \mathbb{C}^2) $ will be referred to as the single-particle Hamiltonian. 
Let us briefly comments on some of its properties:
\begin{enumerate}
\item 
Since $ H_N =: i K_N $ with $ K_N $ real and skew symmetric, the spectrum of $ H_N $ is symmetric about the origin, i.e.,  $\sigma(H_N) = \{ \pm \lambda_1, \pm \lambda_2,  \dots , \pm \lambda_N\}  $ with $ 0\leq \lambda_1 \leq \lambda_2 \leq \dots \leq \lambda_N $ denoting its non-negative eigenvalues. 
\item The unitary transformation $ u := \frac{1}{\sqrt{2}} \begin{pmatrix} 1 & 1 \\ i & -i \end{pmatrix}  $ rotates the spin matrices, $ u^* \sigma^1 u = \sigma^2 $ and $ u^* \sigma^2 u = \sigma^3 $. Through $ U :=\bigoplus_{j=1}^N u $ one may lift this rotation to a local transformation on $ \ell^2(\{1,\dots, N\} ; \mathbb{C}^2)  \simeq  \bigoplus_{j=1}^N  \mathbb{C}^2 $. Under this transformation, the Hamiltonian~\eqref{blockmatrix} turns into a Jacobi matrix block matrix $ U^* H_N U $, in which the variables $ \{ \nu_\xi \} $ are on the diagonal.\\
Performing another change of variables, in which we permute the indices in $ U^* H_N U $, the Hamiltonian is seen to be unitarily equivalent (denoted here by $ \simeq $) to the block matrix
\begin{equation} \label{eq:blockmatrix}
H_N \simeq \begin{pmatrix} -A & - B \\ B & A \end{pmatrix} \quad \mbox{with} \quad A = \begin{pmatrix} - \nu_1 & \mu_1 &  &  \\
				\mu_1 &  \ddots & \ddots &  \\
				&  \ddots & \ddots & \mu_{N-1}\\
				& &\mu_{N-1}  & - \nu_N 
				\end{pmatrix}
\end{equation}
and
\begin{equation}
B =\begin{pmatrix}  0 & \gamma_1\mu_1 &  &  \\
				- \gamma_1\mu_1 &  \ddots & \ddots &  \\
				&  \ddots & \ddots & \gamma_{N-1}\mu_{N-1}\\
				& &- \gamma_{N-1}\mu_{N-1}  & 0
				\end{pmatrix}
\end{equation}
In the isotropic case, i.e. $ \gamma_\xi = 0 $, and if the spin coupling is homogeneous, i.e. $ \mu_\xi = \mu $ for all $ \xi $, the Hamiltonian thus reduces to (two copies of) a discrete 
Schr\"odinger operator on $ \ell^2(\{1,\dots, N\}) $ with hopping $ \mu $ and potential given by $ \{ \nu_\xi \} $. 
\end{enumerate}

To diagonalize $S_N$, we make a Bogoliubov transformation. More precisely, let $O $  be the real orthogonal $ 2N\times 2N $ matrix which brings the skew-symmetric matrix $ K_N $ into its canonical block form, 
\begin{equation}
O K_N O^T= \Lambda = \bigoplus_{j=1}^N \Lambda_j \, , \quad \mbox{where} \quad \Lambda_j := \begin{pmatrix} 0 & \lambda_j \\ - \lambda_j & 0 \end{pmatrix} \, . 
\end{equation}
Regarding this as a change of variables and recalling the quadratic form relation~\eqref{eq:Crep}, it is natural to define
\begin{equation} \label{defbop}
\mathcal{B} := O \mathcal{A} \quad \mbox{and then label} \quad \mathcal{B} = (b_1^+, b_1^- , \cdots, b_{N}^+, b_N^-)^T \, ,
\end{equation}
in analogy to $\mathcal{A}$. The Hamiltonian is then in its canonical form in terms of these $b$-operators:
\begin{equation} \label{diagham}
S_N = \frac{ 1}{2} \,  \mathcal{A}^T H_N \mathcal{A } =   \sum_{j=1}^N \lambda_j \,  i b_j^+ b_j^- \, . 
\end{equation}
Let us summarize some basic properties of these operators:
\begin{enumerate}
\item Since $ O $ is a real orthogonal matrix, the  algebra of Majorana fermions is preserved under this transformation, i.e., $ (b_j^\#)^* = b_j^\# $, $ (b_j^\#)^2 = 1 $ and $ \{ b_j^\#, b_k^\flat \} = 2 \delta_{j,k} \delta_{\#,\flat} $ for both $ \#,\flat \in \{\pm\} $ and all $ j, k \in \{ 1, \dots , N \} $. 
\item For $ j \in \{ 1,\dots ,N\}$, the operators $ i b_j^+ b_j^-  $ are self-adjoint and pairwise commute. Since $ (i b_j^+ b_j^- )^2 = \id $ their eigenvalues are $ \pm 1 $. They measure the individual fermion parity. More precisely, the Fermi creation and annihilation operators 
\begin{equation}
 \psi_j^* :=  \frac{1}{2} \left( b_j^+ - i b_j^- \right) \, , \quad   \psi_j :=  \frac{1}{2} \left( b_j^+ + i b_j^- \right) \, ,
\end{equation}
corresponding to these Majorana modes, satisfy 
\begin{equation}
2 \psi_j^* \psi_j - \idty = i b_j^+ b_j^- \, .
\end{equation} 
The Hamitonian $S_N$ commutes, $ [P_N,S_N] = 0 $, with the total fermion parity
\begin{equation}
P_N := i b_N^+ b_N^-  \cdots  i b_1^+ b_1^- \, .
\end{equation}
Since the latter is self-adjoint $ P_N^* = P_N  $ and satisfies $ P_N^2 = \id $, it also has eigenvalues $ \pm 1 $. 
The orthogonal transformation~\eqref{defbop} preserves the fermion parity operator, i.e. 
\begin{equation}\label{eq:Fermionparitychange}
P_N = \det O \cdot  i a_N^+ a_N^-  \cdots  i a_1^+ a_1^- = \det O \cdot \sigma_N^3 \dots \sigma_1^3 \,  .
\end{equation}
(This follows most easily by restricting wlog to the case $ \det O = 1 $, for which the orthogonal transformation can be implemented on the Hilbert space by a unitary dynamics generated by a quadratic Hamiltonian which commutes with $ P_N $, cf.~\cite{B04}.)
\end{enumerate}
Since the spin Hamiltonian $S_N$ is quadratic and diagonal in $ i b_j^+ b_j^- $, 
a number of important consequences follow:
\begin{enumerate}
\item 
The spectrum of $ S_N $ can be completely described in terms of the joint eigenstates of the collection of the operators $ i b_j^+ b_j^-  $. 
To do so, we start from the unique normalized vector $\Omega \in \mathcal{H}_N$
defined by $ \psi_j \Omega =0$ for all $1 \leq j \leq N$.  Next, for $\alpha = (\alpha_1, \alpha_2, \cdots, \alpha_N) \in \{0,1\}^{N}$, the vectors 
\begin{equation} \label{evecs}
\Psi_{\alpha} = (\psi_1^*)^{\alpha_1} \cdots (\psi_N^*)^{\alpha_N} \Omega 
\end{equation}
form an orthonormal basis of $\mathcal{H}_N$. In fact, they are also eigenvectors of $S_N$: 
\begin{equation} \label{evals}
S_N \Psi_{\alpha} = E_{\alpha} \Psi_{\alpha} \quad \mbox{with} \quad E_{\alpha} =  2 \sum_{j : \alpha_j=1} \lambda_j -E  
\end{equation}
where  $ E = \sum_{j=1}^N \lambda_j $ stands for the negative ground-state energy.
The fermion parity of these eigenstates is
\begin{equation}\label{eq:Fermparityeigen}
P_N \Psi_\alpha = (-1)^{\sum_{j=1}^N \alpha_j + N} \Psi_\alpha 
\end{equation}
for all $ \alpha  \in \{0,1\}^{N} $. The ground-state is unique and given by $ \Psi_{(0,\dots, 0)} $ if and only if $ H_N $ has a trivial kernel.
\item The time evolution is trivial on the $ b $-operators, i.e.
\begin{equation}\label{eq:timeevbops}
\begin{pmatrix}
b_j^+(t) \\ b_j^-(t) 
\end{pmatrix} 
:=  \begin{pmatrix}
\tau_t(b_j^+) \\ \tau_t(b_j^-) 
\end{pmatrix} 
= e^{2 t \Lambda_j} \begin{pmatrix}
b_j^+ \\ b_j^-
\end{pmatrix}   \, . 
\end{equation}
Given this, by setting $\mathcal{A}(t) := \tau_t( \mathcal{A})$, understood component-wise as above, one finds
\begin{equation} \label{aevol}
\mathcal{A}(t) = \tau_t( O^T \mathcal{B}) = O^T e^{2t \Lambda} O \mathcal{A} = e^{-2itH_N} \mathcal{A}  \, .
\end{equation}
\item
The quadratic nature of $ S_N $ implies that any induced thermal state, i.e. $ \rho = e^{- \beta S_N} / \tr  e^{- \beta S_N}  $ with $ \beta \in (0,\infty) $, or any eigenstate of $ S_N $ is quasi-free (cf.~\cite{B04}). The same applies to the functionals which result from these through decorations by the fermion parity operator, in particular, 
\begin{equation}\label{eq:twistedcorr}
\llangle \cdot \rrangle := \tr\left( (\cdot) P_N \rho\right) \big/ \tr\left( P_N \rho\right)
\end{equation} 
assuming $ \tr\left( P_N \rho\right)  \neq 0 $. 
\end{enumerate}

The last observation will be essential in calculating the correlation functions~\eqref{eq:corrfct}.

\subsection{Correlation functions} 
Using the Jordan-Wigner transformation, the spin correlations functions~\eqref{eq:corrfct} can be explicitly expressed in terms of correlation functions involving the $ a $-operators.

\subsubsection{Reduction to Majorana correlations}
As a warm-up, let us first consider all single spin correlations and all those involving $ \sigma^3 $. Recall that we have set
$\mathcal{A}(t) = \tau_t(\mathcal{A}) =: (a_1^+(t), a_1^-(t), \cdots, a_N^+(t), a_N^-(t))^T$ the latter a notation we will use below.
\begin{lemma}\label{lem:3}
Let $\rho$ be a quasi-free state and assume $ [\rho, S_N] = 0 $. One has that
\begin{equation} \label{singlecor}
 \langle \sigma_\xi^1 \rangle  =  \langle \sigma_\xi^2 \rangle = 0 \, , \quad 
  \langle \sigma_\xi^3 \rangle  =  i \langle a_\xi^+a_\xi^-  \rangle  \,  ,
\end{equation}
and 
\begin{align}\label{eq:corrsigma3}
& \langle \tau_t( \sigma_\xi^3) \sigma_\eta^1 \rangle = \langle \tau_t( \sigma_\xi^3) \sigma_\eta^2 \rangle  = 0 \, ,  \notag \\
& \langle \tau_t( \sigma_\xi^3) \sigma_\eta^3 \rangle  -  \langle  \sigma_\xi^3 \rangle \langle \sigma_\eta^3 \rangle  =  \langle a_\xi^+(t)a_\eta^+ \rangle \cdot \langle a_\xi^-(t)a_\eta^- \rangle -  \langle a_\xi^+(t)a_\eta^- \rangle \cdot \langle a_\xi^-(t)a_\eta^+ \rangle \, ,
\end{align}\
for any $1 \leq  \xi, \eta \leq N $ and any $ t \in \mathbb{R} $.
\end{lemma}
\begin{proof}
Inserting the Jordan-Wigner relation~\eqref{cdef} we  obtain for any $\xi$:
\begin{equation} \label{0exp}
 \langle \sigma_{\xi}^1 \rangle  = i^{\xi -1} \left\langle \left( \prod_{\ell = 1}^{\xi -1} a_{\ell}^+ a_{\ell}^- \right)  a_{\xi}^+  \right\rangle =0 
\end{equation}
the final equality follows as these states are quasi-free. In fact, the expectation of the product is then a pfaffian, and there
are an odd number of $a$-operators. The result for $\sigma_{\xi}^2$ is similar. The third identity in~\eqref{singlecor} immediately follows from $ \sigma_\xi^3 = i a_\xi^+a_\xi^- $.

For a derivation of~\eqref{eq:corrsigma3} we proceed similarly using~\eqref{cdef} and the fact that the state $ \rho $ is quasi-free:
\begin{equation} \label{0exp2}
\langle \tau_t(\sigma_{\xi}^3) \sigma_{\eta}^1 \rangle  = i^{\eta} \left\langle a_{\xi}^+(t) a_{\xi}^-(t)  \left( \prod_{\ell = 1}^{\eta -1} a_{\ell}^+ a_{\ell}^- \right) a_{\eta}^{+}  \right\rangle =0 \, ,
\end{equation}
since the number of $a$-operators is odd. 
The result for $\sigma_{\eta}^2$ is again argued similarly.

To evaluate the remaining correlation, observe that
\begin{align}
\langle \tau_t( \sigma_\xi^3) \sigma_\eta^3 \rangle  -  \langle  \sigma_\xi^3 \rangle \langle \sigma_\eta^3 \rangle & =  -   \langle a_\xi^+(t) a_\xi^-(t)a_\eta^+a_\eta^- \rangle +  \langle a_\xi^+ a_\xi^- \rangle \cdot \langle a_\eta^+ a_\eta^- \rangle \notag \\
& =   \langle a_\xi^+(t)a_\eta^+ \rangle \cdot \langle a_\xi^-(t)a_\eta^- \rangle -  \langle a_\xi^+(t)a_\eta^- \rangle \cdot \langle a_\xi^-(t)a_\eta^+ \rangle
\end{align}
where the last equality again follows from the fact that $ \rho $ is quasi-free. In this case, the four-point function $ \langle a_\xi^+(t) a_\xi^-(t)a_\eta^+a_\eta^- \rangle  $ reduces to 
a simple pfaffian which can be evaluated e.g. according to Wick's rule~\eqref{def:quasifree}. Moreover, by time invariance we have
$ \langle a_\xi^+(t)a_\xi^-(t) \rangle= \langle a_\xi^+a_\xi^- \rangle$. 
\end{proof}

The $ \sigma^3 $-correlation~\eqref{eq:corrsigma3} is readily seen to decay in the distance $ |\xi-\eta| $ whenever the two-point functions involving the $a $-operators 
are known to do so. To establish a similar result for the correlations in the $12$-plane, we again start from the Jordan-Wigner transformation~\eqref{cdef} and write for $ w, w\in \{ 1,2\} $: 
\begin{equation}\label{eq:12cor1}
\langle \tau_t( \sigma_\xi^w) \, \sigma_\eta^{w'} \rangle =  (-1)^{w+w'}  \langle \left( \prod_{\ell=1}^{\xi-1} i a_\ell^+(t) a_\ell^-(t) \right) a^{\#_w}_\xi(t) a^{\#_{w'}}_\eta \left( \prod_{m=1}^{\eta-1} i a_\ell^+ a_\ell^- \right) \rangle \, , 
\end{equation}
where we introduced the abbreviation:
\begin{equation}
  \#_w := \begin{cases} + & w=1 \\ - & w= 2 \end{cases} \, . 
 \end{equation}
The above average does not quite fit our needs if one aims to apply Theorem~\ref{thm:pf} or  \ref{thm:majodis}. We therefore rewrite the product
$$  i  a^+_{1} a^-_{1} \cdots i a^+_{\eta-1} a^-_{\eta-1}  = (\det O) \cdot  i a^+_{\eta} a^-_{\eta} \dots i a^+_{N} a^-_{N} \, P_N \, , 
$$
using the identity~\eqref{eq:Fermionparitychange} for the fermion parity. This brings the  twisted average~\eqref{eq:twistedcorr} into the equation: 
\begin{align} \label{eq:12cor}
 \langle \tau_t( \sigma_\xi^w) \, \sigma_\eta^{w'} \rangle 
 = & (-1)^{w+w'}  \;  i^{\xi+ N-\eta} (\det O) \left(\tr P_N \rho\right) (-1)^{w'-1} \\
& \mkern10mu \times \llangle  a^+_{1}(t) a^-_{1}(t) \cdots a^+_{\xi-1}(t) a^-_{\xi-1}(t) \ a_\xi^{\#_w}(t) a_\eta^{\flat_{w'}} a^+_{\eta+1} a^-_{\eta +1} \dots a^+_{N} a^-_{N}  \rrangle \, ,   \notag 
\end{align}
for all $1 \leq  \xi \leq \eta \leq N $, any $ t \in \mathbb{R} $, and $ w, w' \in \{ 1,2\} $. Here we have set $\flat_{w'} = - \#_{w'}$.
In case~\eqref{eq:twistedcorr} defines a quasi-free functional, the last expression ($ \llangle \dots \rrangle $) is the pfaffian
\begin{equation}\label{eq:XYPfaffian}
  \pf\left(\llangle   a_{x_j}^{\#_j}(t_j)  a_{x_k}^{\#_k}(t_k) \rrangle\right)_{1\leq j < k \leq 2n }  
\end{equation}
where $  x = (1,1, 2,2 , \dots ,\xi-1,\xi-1, \xi , \eta , \eta +1 , \eta +1, \dots ,N, N ) $, the corresponding vector of signs is $ (+,-, \dots, +,- , \#_w, \flat_{w'} , +, - , \dots , +, -) $, and  $ n = N + \xi - \eta $. 
This is precisely the setting of Theorem~\ref{thm:pf} or \ref{thm:majodis} with distance given by
\begin{equation} \label{configdist}
r(x) = \max_{j } | x_{2j-1} - x_{2j} | = | \xi - \eta | \, .
\end{equation}
In order to apply these theorems, it remains to determine two-point functions associated with the $ a $-operators.

\subsubsection{Calculating Majorana correlations}
 All the relevant information concerning the spin correlations of interest is encoded in the following $2N\times 2N $ matrix:
\begin{align}
\Gamma^{\mathcal{A}}(t,s) & :=   \langle \mathcal{A}(t) \mathcal{A}(s)^T \rangle  \, , \quad t,s \in \mathbb{R} \, .
\end{align}
If $ \rho $ commutes with $ S_N $, then this correlation matrix only depends
on the time difference, 
$ \Gamma^{\mathcal{A}}(t,s) = \Gamma^{\mathcal{A}}(t-s,0) $.
Using (\ref{aevol}), it is clear that
\begin{equation} \label{athrub}
\Gamma^{\mathcal{A}}(t,0) = \langle e^{-2itH_N} \mathcal{A}(0) \mathcal{A}(0)^T \rangle = e^{-2itH_N} O \Gamma^{\mathcal{B}}(0,0) O^T
\end{equation}
where we have similarly set $\Gamma^{\mathcal{B}}(t,s) :=   \langle \mathcal{B}(t) \mathcal{B}(s)^T \rangle$. We need only
determine the static $b$-correlations and this is the content of the following 
\begin{lemma}\label{lem:Bcormat}
Assume either
\begin{description}
\item[Case 1:] $ \rho = e^{-\beta S_N} / \tr e^{-\beta S_N} $ with $ \beta >0 $, or
\item[Case 2:] $ H_N $ has simple spectrum and $ \rho = | \Psi_\alpha \rangle \langle \Psi_\alpha | $ with $ \alpha \in \{0, 1\}^N $, or
\item[Case 3:] $ H_N $ has a trivial kernel and $ \rho = P_N e^{-\beta S_N} / \tr P_N e^{-\beta S_N} $ with $ \beta >0 $.
\end{description}
Then for any $ t, s \in \mathbb{R} $:
\begin{equation} \label{Bcormat}
\Gamma^{\mathcal{A}}(t,s) = e^{-2i(t-s) H_N} \,   f_\rho(H_N)   \, ,
\end{equation}
where $f_{\rho}: \R \to  \mathbb{R} $ is the function given by:
\begin{description}
\item[Case 1:]  $ f_\rho(\lambda) = 2 ( 1+ e^{-2\beta \lambda} )^{-1}  $.
\item[Case 2:]  $ f_\rho(\lambda)  = 2 \chi_{\Delta_\alpha}(\lambda) $\\[1ex] with $ \chi_{\Delta_\alpha} $ the characteristic function onto the set
$\Delta_{\alpha} = \{ \lambda_j | \alpha_j =0 \} \cup \{ - \lambda_j | \alpha_j = 1 \} $. 
\item[Case 3:]  $ f_\rho(\lambda) = 2 ( 1- e^{2\beta \lambda} )^{-1}   $.
\end{description}
\end{lemma}
\begin{proof}
Given (\ref{athrub}) and the fact that $ K_N = O^T \Lambda O $, the claim is equivalent to showing that
\begin{equation}
\langle \mathcal{B} \mathcal{B}^T \rangle_{\rho} = \,  f_\rho(i \Lambda)  \, .
\end{equation}
One easily checks that, in the cases considered above, the off-diagonal expectations are zero, i.e., $ \langle b_j^\# b_k^\flat \rangle_{\rho} = 0 $ for $ j \neq k $, and so
\begin{equation}
\langle \mathcal{B} \mathcal{B}^T \rangle_{\rho} = \bigoplus_{j=1}^N  \big\langle  \begin{pmatrix} b_j^+ \\ b_j^- \end{pmatrix}  \begin{pmatrix} b_j^+ & b_j^- \end{pmatrix} \big\rangle_{\rho} 
 = \bigoplus_{j=1}^N \left[ \id -i  \langle b_j^+ b_j^- \rangle_{\rho} \begin{pmatrix} 0 & i \\ -i & 0 \end{pmatrix}  \right] \, .
\end{equation}
Here the last line results from explicit matrix multiplication using $ (b_j^\# )^2 = 1 $ and $ b_j^- b_j^+ = - b_j^+b_j^- $. It thus remains to calculate $ i  \langle b_j^+ b_j^- \rangle_\rho $ in the cases mentioned above.

\begin{enumerate}
\item 
In case of a thermal state,  $ \rho = e^{-\beta S_N} / \tr e^{-\beta S_N} $, we have 
\begin{align}
i \tr \left( b_j^+ b_j^-  e^{-\beta S_N} \right) & = \left( \sum_{n_j \in \{0,1\}} e^{-\beta \lambda_j (2n_j-1)} (2n_j-1) \right)\prod_{k\neq j}^N \left( \sum_{n_k \in \{0,1\}} e^{-\beta \lambda_k (2n_k-1)} \right) \notag \\
& =  - \tanh(\beta \lambda_j) \, \tr  e^{-\beta S_N}  \, . 
 \end{align} 
This implies $  -i \langle b_j^+ b_j^- \rangle_\rho \begin{pmatrix} 0 & i \\ -i & 0 \end{pmatrix}= \tanh(i \beta \Lambda_j) $.
\item
In case of an eigenstate,  $ \rho = | \Psi_\alpha \rangle \langle \Psi_\alpha | $, we have
\begin{equation}
i \langle b_j^+ b_j^- \rangle_{\rho} = \langle \Psi_\alpha , (2 \psi_j^* \psi_j - 1) \Psi_\alpha \rangle = \, 2 \alpha_j - 1  \, . 
\end{equation}
Hence $ i \langle b_j^+ b_j^- \rangle_\rho \begin{pmatrix} 0 & i \\ -i & 0 \end{pmatrix} = (2 \alpha_j - 1) \sgn(i \Lambda_j) $. Here that we require $ \lambda_j \neq 0 $ for all $ j \in \{1, \dots , N \} $, which due to the symmetry of the spectrum of $ H_N $ is implied by the simplicity of the eigenvalues. From this the claim follows by distinguishing the cases $ \alpha_j \in \{ 0, 1\} $. Note that the fact that $ H_N $ has simple spectrum implies that $ \chi_\alpha(H_N) $ is well-defined.
\item
In case $ \rho =P_N  e^{-\beta S_N} / \tr P_N e^{-\beta S_N} $, the calculation proceeds similarly to the first case. Since $ \tr P_N e^{-\beta S_N} = (-2)^N \prod_{j=1}^N \sinh(\beta \lambda_j) $, we again need the assumption that $ H_N $ has a non-trivial kernel. 
\end{enumerate}
\end{proof}

Before turning to our main result, let us conclude this section with some historical remarks.
Lemma~\ref{lem:3} and \ref{lem:Bcormat} together with \eqref{eq:12cor1} or \eqref{eq:12cor} yields a general expression for the time-dependent spin correlation functions of the $ XY $ model in terms of pfaffians involving the single-particle Hamiltonian $ H_N $ entering
\begin{equation}\label{eq:singlecor}
\langle a_\xi^\#(t)a_\eta^\flat(s) \rangle = \langle \delta_\xi^\# , e^{-2i (t-s) H_N} f_\rho(H_N) \delta_\xi^\flat \rangle \, ,
\end{equation} 
where $ \{ \delta_\xi^\# \} $ denotes the canonical orthonormal basis in $ \ell^2(\{ 1, \dots , N \};\mathbb{C}^2 ) $. 
The fact that spin correlations in the $ XY $-chain are expressible in terms of such pfaffians (or determinants) is an observation which dates back to the seminal paper \cite{LM} for the time-independent case. In the homogeneous case ($\mu_\xi = \mu $) these explicit expressions are used to show that the ground-state correlations in the $12 $-direction exhibit an algebraic fall-off -- a fact which should be contrasted to the exponential decay~\ref{eq:explocthm} below in the presence of an additional random field $ \{\nu_\xi \} $.

Explicit expressions for the time-dependent correlation functions go back to \cite{Nie} (for the $ 3 $-direction) and \cite{CBA,BJ76} (for the $ 12 $-direction). They have been 
the starting point for numerous further studies (see, e.g.~\cite{SNM}).

\subsection{Main result: dynamical localization} 
Our main result in this section concerns the case in which the one-particle Hamiltonian $ H_N $ is random and can be proven to exhibit strong-dynamical localization (cf.~\eqref{eq:expdecaysl}). A standard example of a random version of $ H_N $ is the case that the spin coupling parameters $ \{ \mu_\xi \} $ and $ \{ \gamma_\xi \} $ are constant and the external magnetic field $ \{ \nu_\xi \} $ forms iid random variables. We will discuss the applicability of the following general theorem in this case below.
\begin{theorem}[Strong dynamical localization in spin chain]\label{thm:XY0}
Suppose that the single-particle Hamiltonian $ H_N $ associated with the spin-chain $ S_N $ is a random operator on $ \ell^2(\{1,\dots,N\};\mathbb{C}^2) $ which for all $ N \in \mathbb{N} $ satisfies:
\begin{enumerate}
\item $ H_N $ has almost-surely simple spectrum.
\item the eigenfunction correlator of $H_N$ exhibits complete strong dynamical localization in the sense that for all $ \xi,\eta \in \{1,\dots , N\} $:
\begin{equation}\label{eq:complloc0}
\sup_{\#,\flat \in \{\pm\} } \mathbb{E}\left[ \sup_{\substack{f \in L^\infty(\mathbb{R}) \\ \| f \|_\infty \leq 1 }}  \left| \langle \delta_{\xi}^\#  , f(H_N)  \, \delta_{\eta}^\flat \rangle \right| \right] \leq C \, e^{-\mu |\xi - \eta | }
\end{equation}
with some $ N $ independent constants $ C, \mu \in (0,\infty) $.
\end{enumerate} 
Then the time-dependent  spin correlations associated to either of the states corresponding to \eqref{def:state} also exhibit strong dynamical localization in the sense that there are some $ C', \mu' \in (0,\infty) $ for which, given any $ N \in \mathbb{N} $:
 \begin{equation}
\mathbb{E} \left[ \sup_{t \in \mathbb{R}} \left| \langle \tau_t(\sigma_\xi^w) \sigma_\eta^{w'} \rangle - \langle \sigma_\xi^w \rangle \langle \sigma_\eta^{w'} \rangle \right| \right] \leq C' \, e^{-\mu' |\xi-\eta|} \, . 
\end{equation}
for all $ w,w' \in \{1,2,3\} $ and all $ \xi, \eta \in \{ 1,\dots , N \} $.
\end{theorem}
\begin{proof}
In case $ w= 3 $ or $w' = 3$, the claim immediately follows from Lemma~\ref{lem:3} and~\eqref{eq:singlecor} with $ f_\rho \in L^\infty $ bounded by $ \| f_\rho \|_\infty = 2 $, cf.\ Lemma~\ref{lem:Bcormat}. In this case, the only non-trivial correlation is
\begin{align}\label{eq:33corr}
\mathbb{E} \left[ \sup_{t \in \mathbb{R}} \left| \langle \tau_t(\sigma_\xi^3) \sigma_\eta^{3} \rangle - \langle \sigma_\xi^3 \rangle \langle \sigma_\eta^{3} \rangle \right| \right] 
\ & \leq \ 2 \sup_{\#,\flat \in \{\pm\} } 
\mathbb{E} \left[  \sup_{t \in \mathbb{R}} \left| \langle \delta_\xi^\# , e^{-2it H_N} f_\rho(H_N) \delta_\eta^\flat \rangle \right| \right] \notag \\
& \leq \ 4 \,  C e^{-\mu |\xi - \eta|} \, . 
\end{align}
In case $ w, w' \in \{1,2\} $, we first restrict the discussion to the case of eigenstates, $ \rho = |\Psi_\alpha\rangle\langle \Psi_\alpha | $ and envoke the representation~\eqref{eq:12cor}. 
Given~\eqref{eq:Fermparityeigen}, the prefactor in~\eqref{eq:12cor} is bounded by one. One thus has for $ w,w' \in \{1,2\} $:
\begin{align}
 \left| \langle \tau_t(\sigma_\xi^w) \sigma_\eta^{w'} \rangle_\alpha  \right|  
= \left| \langle a^+_{1}(t) a^-_{1}(t) \cdots a^+_{\xi-1}(t) a^-_{\xi-1}(t) \ a_\xi^{\#_w}(t) a_\eta^{\flat_{w'}} a^+_{\eta+1} a^-_{\eta +1} \dots a^+_{N} a^-_{N} \rangle_\alpha \right|  \, . 
\end{align}
Since eigenstates are quasi-free, the right-hand side is the pfaffian~\eqref{eq:XYPfaffian} (with $ \llangle \cdot \rrangle $ replaced by $ \langle \cdot \rangle_\alpha $). The claim thus follows from Theorem~\ref{thm:majodis} using 
\begin{equation}\label{eq:unifcorbound}
\mathbb{E} \left[ \sup_{\alpha} \sup_{t \in \mathbb{R}} \left| \langle \delta_\xi^\# , e^{-2it H_N} 2 \chi_{\Delta_\alpha}(H_N) \delta_\eta^\flat \rangle \right| \right] \leq 2 \,  C e^{-\mu |\xi - \eta|} \, , 
\end{equation}
for all $ \xi , \eta $. Note that (\ref{eq:unifcorbound}) follows from Lemma~\ref{lem:Bcormat} and assumption~\eqref{eq:complloc0}. Moreover, as indicated in (\ref{configdist}), the distance of the configuration of Majorana fermions entering the pfaffian~\eqref{eq:XYPfaffian} is $ |\xi-\mu| $. 

In case of thermal states, $ \rho = e^{-\beta S_N} / \tr e^{-\beta S_N } $, the result in case $ w,w' \in \{1,2\} $ follows from the above, since 
\begin{equation}
	 \left| \langle \tau_t(\sigma_\xi^w) \sigma_\eta^{w'} \rangle_\beta \right| \leq \ \sup_\alpha \left| \langle \tau_t(\sigma_\xi^w) \sigma_\eta^{w'} \rangle_\alpha\right|
\end{equation}
The claimed bound is hence a consequence of Theorem~\ref{thm:majodis} with the help of~\eqref{eq:unifcorbound} and taking the first remark below  Theorem~\ref{thm:majodis} into account.
\end{proof}

Several remarks apply:
\begin{enumerate}
\item As was shown in~\cite[Prop. A.1]{AS}, $S_N$ and hence $ H_N $ has simple spectrum for Lebesgue-almost all $ \{ \nu_\xi \} \in \mathbb{R}^N $. Taking $  \{ \nu_\xi \} $ independently distributed random variables with a single-site distribution which is absolutely continuous hence implies that $ H_N $ has almost surely simple spectrum. (Since the latter is symmetric about the origin, this in particular implies that the kernel of $ H_N $ is trivial almost surely.)
\item As was explained in Subsection~\ref{ss:JW}, in the isotropic case ($ \gamma_\xi = 0 $) and for homogeneous spin coupling ($ \mu_\xi = \mu $), the Hamiltonian $ H_N $ reduces to (two copies of) the Anderson model with random potential $ \{ \nu_\xi \} $. In this case, strong dynamical localization in the sense of~\eqref{eq:complloc0} is known to occur for iid random variables under fairly general conditions on the single-site distribution (cf.~\cite{AW} and references therein). 

In the non-isotropic, but homogeneous case ($\gamma_\xi = \gamma$ and $ \mu_\xi = \mu$) a less complete picture is generally available.  A result in~\cite{ESS} covers the regime of large disorder in case $ \{ \nu_\xi \} $ are iid with absolutely continuous distribution with a compact support.  In \cite{CS} strong dynamical localization~\eqref{eq:complloc0} is established for strong enough spin coupling $ |\mu |$.  

\end{enumerate}

Theorem~\ref{thm:XY0} applies to all eigenstates and the thermal states. If one just aims at a localization statement concerning thermal states or ground-state, without any dynamics, less has to be assumed.
\begin{theorem}[Localization for thermal states or the ground-state]\label{thm:XY1}
Suppose that the single-particle Hamiltonian $ H_N $ associated with the spin-chain $ S_N $ is a random operator on $ \ell^2(\{1,\dots,N\};\mathbb{C}^2) $ which for all $ N \in \mathbb{N} $ satisfies:
\begin{enumerate}
\item $ H_N $ has almost-surely a trivial kernel.
\item the Green function of $H_N$ at zero exhibits fractional moment localization in the sense that for some $ s \in (0,1)$ and all $ \xi,\eta \in \{1,\dots , N\} $:
\begin{equation}\label{eq:complloc}
\max_{\#,\flat \in \{\pm\} } \sup_{\gamma \in \mathbb{R}} \, \mathbb{E}\left[ \left| \langle \delta_{\xi}^\#  , (H_N - i\gamma)^{-1}  \delta_{\eta}^\flat \rangle \right|^s \right] \leq C \, e^{-\mu |\xi - \eta | }
\end{equation}
with some $ N $ independent constants $ C, \mu \in (0,\infty) $.
\end{enumerate} 
Then the thermal spin correlations exhibit  localization in the sense that there is $ C', \mu' \in (0,\infty) $ for which, given any $ N \in \mathbb{N} $:
 \begin{equation}\label{eq:explocthm}
\mathbb{E} \left[  \left| \langle \sigma_\xi^w \sigma_\eta^{w'} \rangle_\beta - \langle \sigma_\xi^w \rangle_\beta \langle \sigma_\eta^{w'} \rangle_\beta \right| \right] \leq C' \max\{ 1, \beta^{-s} \}  \, e^{-\mu' |\xi-\eta|} \, 
\end{equation}
for all $ w,w' \in \{1,2,3\} $, all $ \xi, \eta \in \{ 1,\dots , N \} $, and  all  $ \beta \in (0,\infty] $ . 
\end{theorem}
\begin{proof}
We will only give a proof in case $ \beta \in (0,\infty) $ since the ground-state case $ \beta = \infty $ follows by a limiting argument.

In case $ w= 3 $ or $w' = 3$, we proceed as in the proof of Theorem~\ref{thm:XY0}. In particular, in the only non-trivial case $ w=w'=3 $, we use the first estimate in~\eqref{eq:33corr}
in which $ f_\rho(\lambda) = 2 (1+e^{-2\beta\lambda} )^{-1} $. Using an argument from~\cite{AG98}, we may write for $ \epsilon_\beta = \lceil \frac{\beta}{\pi} \rceil \frac{\pi}{\beta}  \in [ 1, 1+\pi/\beta) $ and all $ \lambda \in \mathbb{C}\backslash \{   \frac{i \pi n}{2\beta} \, | \, n \in \mathbb{Z}  \, \mbox{odd} \} $ with $ | \Im \lambda | <  \epsilon_\beta $:
\begin{align}\label{eq:expag}
 f_\rho(\lambda) =& \  2 Q_1(\lambda) + 2 Q_2(\lambda) \quad \mbox{with} \notag \\
&  Q_1(\lambda) := \frac{1}{2\beta} \sum_{\substack{n \in \mathbb{Z}  \, \mbox{odd} \\  \frac{\pi |n|}{2\beta} < \epsilon_\beta} } \left(  \frac{i \pi n}{2\beta} - \lambda \right)^{-1} \, , \\
&   Q_2(\lambda) := \int_{-\infty}^\infty   f_\rho(u) \left[ \frac{1}{u-i\epsilon_\beta-\lambda} - \frac{1}{u+i\epsilon_\beta-\lambda}\right] \, du \, . 
\end{align}
The contribution $ \left| \langle \delta_\xi^\# ,Q_2(H_N) \delta_\eta^\flat \rangle \right| $ is estimated with the help of a Combes-Thomas bound, cf.~\eqref{eq:CombesThomas} and~\cite{AG98}. The remaining contribution  is estimated using~\eqref{eq:complloc}. More explicitly:
\begin{align}
2 \, \mathbb{E}\left[ \left| \langle \delta_\xi^\# ,Q_1(H_N) \delta_\eta^\flat \rangle \right| \right] \ & \leq \frac{1}{\beta} \sum_{\substack{n \in \mathbb{Z}  \, \mbox{odd} \\  \frac{\pi |n|}{2\beta} < \epsilon_\beta} } \frac{(2\beta)^{1-s}}{\pi^{1-s} |n|^{1-s} }\mathbb{E}\left[ \big| \langle \delta_\xi^\#, \big(  H_N - \frac{i \pi n}{2\beta} \big)^{-1}  \delta_\eta^\flat \rangle \big|^s\right] \notag \\
& \leq \frac{2^{2-s} C }{\pi^{1-s} \beta^s } \, e^{-\mu |\xi-\eta|} \, \sum_{n=1}^{2 \lceil \frac{\beta}{\pi} \rceil} \frac{1}{n^{1-s}} \, , 
\end{align}
where the last sum is bounded by a constant times $ \beta^s $ for all $ \beta \geq \pi  $.

In the other cases $ w, w' \in \{ 1,2\} $ we rewrite using~\eqref{eq:12cor}:
\begin{align}
\left| \langle \sigma_\xi^w \sigma_\eta^{w'} \rangle_\beta \right|   = & \  | \tr P_N e^{-\beta S_N} | \,  \\
& \times \left|  \llangle  a^+_{1} a^-_{1} \cdots a^+_{\xi-1} a^-_{\xi-1} \ a_\xi^{\#_w} a_\eta^{\flat_{w'}} a^+_{\eta+1} a^-_{\eta +1} \dots a^+_{N} a^-_{N} \rrangle_\beta \right| \, . 
\notag
\end{align}
The arising pfaffian~\eqref{eq:XYPfaffian} satisfies the requirements of Therem~\ref{thm:majodis} with $ M_0^{-1} =  | \tr P_N e^{-\beta S_N} | \neq 0 $. To verify the other assumption~\eqref{ass:majdis} in this theorem, we note that by Lemma~\ref{lem:Bcormat} 
\begin{equation}
\llangle a^\#_{\xi} a_\eta^\flat \rrangle_\beta =  \langle \delta_\xi^\# , g_\rho(H_N) \delta_\eta^\flat \rangle \, , \quad g_\rho(\lambda) = 2 (1-e^{2\beta\lambda})^{-1} \, .
\end{equation}
Since $  g_\rho(\lambda) =  f_\rho(-\frac{i\pi}{2\beta} -\lambda) $ we may use~\eqref{eq:expag} together with the fact that $ \pi/(2\beta) < \epsilon_\beta $ and that the kernel of $ H_N $ is trivial to rewrite
\begin{align}\label{eq:expag2}
 g_\rho(H_N)   =& \  2 \, Q_1\left(-\frac{i\pi}{2\beta} -H_N\right)  + 2\,  Q_2\left(-\frac{i\pi}{2\beta} -H_N\right) \, . 
\end{align}
The contribution of the second terms is again bounded using the Combes-Thomas estimate from~\cite{AG98}. For its application note that $ \epsilon_\beta - \pi/(2\beta) \geq 1/2 $. The first term is estimated similarly as above.
\end{proof}

Some remarks:

\begin{enumerate}
\item Quite generally, it is known that thermal states associated to one-dimensional, many-body quantum lattice
systems satisfy exponential decay of correlations, or exponential clustering. In fact, Araki showed \cite{A0} that analyticity arguments
allow one to use Ruelle's classical transfer matrix methods, see e.g. \cite{R}, to prove that the Gibbs state of (e.g. finite range) 
one-dimensional systems satisfy exponential clustering at any positive temperature. Consequently, Araki's result yields, deterministically,
exponentially decaying bounds on thermal states of the XY-model. By contrast, our averaged bounds, in this random setting, are not only more explicit,
they are also uniform, in the sense that they survive the $\beta \to \infty$ limit. 

\item Some previous results concerning decay of correlations in random XY-models exist. In \cite{KP}, some bounds 
on static correlations of certain observables in the ground state of the isotropic XY-model are considered. 
 In \cite{HSS}, a bound on averaged, static ground state correlations  
is proven, again for the random isotropic XY-model. More precisely, for a chain of length $N \geq 1$, a bound of the form
\begin{equation}
\mathbb{E} \left( | \langle AB \rangle - \langle A \rangle \langle B \rangle | \right) \leq C N \| A \| \| B \| e^{- \mu |\xi - \eta|} \quad \mbox{for all } A \in \mathcal{A}_{\{ \xi\} }, B \in \mathcal{A}_{\{ \eta \} }
\end{equation} 
is obtained by combining a zero-velocity Lieb-Robinson bound, a Lifshitz tails estimate for the Anderson model, and well-known methods, see e.g. \cite{NS},
for deriving correlation decay in the ground state of gapped many-body systems. (Suggestions for improvement of this method can be found in~\cite{E}.)

\end{enumerate}

\section{Upper Bounds on Certain Bordered Determinants and Pfaffians}

The technical core of this paper, on which the proofs of our main results rest, are two estimates on certain bordered determinants and pfaffians.
\subsection{Determinants}
The following is the main new technical  result for determinants.
\begin{theorem} \label{Had2.0}
Consider a complex matrix $M\in \mathbb{C}^{(n+1)\times(n+1)}$ with the following block structure
\begin{equation} \label{defM}
M = \left( \begin{matrix}  \alpha & v_1^T & v_2^T \\
w_1 & A & B \\
w_2 & C & D \end{matrix} \right) 
\end{equation}
with $ \alpha \in \mathbb{C} $, column vectors $ v_1 ,w_1 \in \mathbb{C}^p $ and $v_2 , w_2 \in \mathbb{C}^q $, and blocks $ A \in \mathbb{C}^{p\times p} $, $ B \in   \mathbb{C}^{p\times q} $, $ C \in  \mathbb{C}^{q\times p}  $, and $ D \in  \mathbb{C}^{q\times q} $ with $ p+q = n $. 
	If  $ \| M \| \leq 1 $, then we have that
	\begin{equation}\label{eq:MainDetBound}
		\left| \det M \right| \leq  | \alpha| + \| v_2 \| + \| w_1 \| + \|B \| + 2 \sqrt{ \| v_1 \|( \| w_1\| + \| B \|)}
	\end{equation}
\end{theorem} 
While this bound is not sharp, it also does not result from a straightforward application of
Hadamard's inequality \cite{Muir} which asserts that 
\begin{equation}\label{eq:Hadamar}
| \det M |  \leq \prod_{j=1}^{n+1} \| r_j( M) \| \leq \|M \|^n \, \min_j \| r_j( M) \|  \, . 
\end{equation} 
Here $ r_j(M) $ denote the row vectors of the matrix and the last inequality follows since the Euclidean norm of any row is bounded by the matrix norm. 
Since we deal with matrices satisfying $ \| M \| \leq 1 $, the determinant is then bounded by $ \min_j \| r_j( M) \|  $. 
None of the row vectors, in general, have a norm comparable with the right side of~\eqref{eq:MainDetBound}. 
In addition to Hadamard's inequality~\eqref{eq:Hadamar}, the proof of Theorem~\ref{Had2.0} is based on a change of basis and the invariance of the determinant under row operations. 
\begin{proof}[Proof of Theorem~\ref{Had2.0}]
Let $ U  \in  \mathbb{C}^{p\times p}  $ be a unitary transformation for which 
\begin{equation}\label{def:U}
(U v_1)^T = (0, \dots, 0 , \| v_1 \| ) 
\end{equation}
 and take $V \in  \mathbb{C}^{p\times p}  $ to be a unitary which transforms $ A U^T $ into an upper triangular matrix:
\begin{equation}\label{def:V}
V A U^T =  \left( \begin{matrix}  \alpha_{1}  &  \ldots  &\\
&  \ddots    & \vdots \\
  \mbox{} & &   \alpha_{p} \end{matrix} \right) 
\end{equation}
with some $ \alpha_{j} \in \C $. Lifting these matrices 
\begin{equation}\label{eq:UV}
	 \mathcal{U} :=  \left( \begin{matrix} U & 0 \\ 0 & \mathbbm{1}_q \end{matrix} \right) \in \mathbb{C}^{n\times n} \, , \quad  \mathcal{V} :=  \left( \begin{matrix} V & 0 \\ 0 & \mathbbm{1}_q \end{matrix} \right) \in \mathbb{C}^{n\times n} \,  
	\end{equation}
we have
\begin{equation}\label{eq:utrafo}
\tilde{M} := \left( \begin{matrix} 1 & 0 \\ 0 &  \mathcal{V}  \end{matrix}\right) M \left( \begin{matrix}   1 & 0  \\ 0 &  \mathcal{U}^T \end{matrix}\right) = \left( \begin{matrix}  \alpha & (U v_1)^T & v_2^T \\
V w_1 & V A U^T & V B \\
w_2 & C U^T  & D \end{matrix} \right)  \, .
\end{equation}
 Since the unitary transformations leave the norm as well as the modulus of the determinant invariant, 
an application of Hadamard's inequality~\eqref{eq:Hadamar} yields
\begin{equation}
 | \det M | = | \det \tilde{M}|  \leq   \min_j \| r_j ( \tilde{M})\|     \, .
 \end{equation}
We now distinguish between two cases with a variational parameter $ \varepsilon > 0 $.

\begin{description}
\item[{\bf Case 1:}] 
Suppose that $ |\alpha_{p} | \leq \varepsilon $.
The norm of the $(1+p)$th row of $ \tilde{M} $ can then be estimated as
\begin{align}\label{Case1}
\| r_{p+1}(\tilde{M})  \| =  \  \sqrt{ | (Vw_1)_p|^2 +  |\alpha_{p} |^2 + \| r_p(VB) \|^2 } 
 \leq  \| w_1 \| + \varepsilon + \| B \| 
\end{align}
where we used the fact that $ \| r_j(VB) \| \leq \| V B \| \leq \| B \|  $. Thus we have $ | \det \tilde{M}|  \leq  \| w_1 \| + \varepsilon + \| B \|  $. 
\item[{\bf Case 2:}] Suppose that $ |\alpha_{p} | > \varepsilon $.
We then use row operations (which leave the determinant of $ \tilde{M} $ invariant) to eliminate the non-zero entry in $Uv_1$ from the first
row.
The other components in the first row are then modified as follows:
\begin{align} \label{eq:rowop}
\alpha \mapsto \alpha' & := \alpha - s\,   \left( Vw_1\right)_p \notag \\
v_2\mapsto v'_2 & :=  v_2 -  s\,   r_p( VB) \, , \quad \mbox{with} \; s:=\frac{\| v_1 \| }{\alpha_{p}}  .
\end{align}
As a consequence, the norm of this modified first row can be estimated by
\begin{align}
\sqrt{|\alpha'|^2 + \left\| v'_2\right\|^2}   \leq &  \  | \alpha ' | + \| v_2' \|  \leq   | \alpha| +  |s| \  \| Vw_1\|  + \| v_2 \|  +   |s|\ \|  r_p( VB)\| \nonumber \\
\leq &  \ | \alpha| + \frac{1}{\varepsilon} \, \| v_1\| \| w_1\| + \| v_2 \| + \frac{\| v_1 \|}{\varepsilon} \| B\|
\end{align}
Therefore
\begin{equation}\label{Case2}
\left| \det \tilde{M} \right| \leq  \min_j \| r_j( \tilde{M}) \| \leq | \alpha | + \| v_ 2 \| + \frac{\| v_1 \| }{\varepsilon} \left( \|w_1 \| + \| B \|\right) \,  .
\end{equation}
\end{description}

Summarizing, the sum of the left sides of~\eqref{Case1} and \eqref{Case2} constitute an upper bound on $ |\det M | $. Optimizing over $  \varepsilon > 0 $, i.e., taking $\varepsilon = \sqrt{ \| v_1 \|( \| w_1\| + \| B \|)}$, we then arrive at the bound claimed in (\ref{eq:MainDetBound}).
\end{proof}

\subsection{Pfaffians}

Since pfaffians belong to a less popular branch of linear algebra, let us start this subsection by reviewing some basic facts which will be of relevance. (For proofs and much more, see \cite{Muir}.)

Let $M \in \mathbb{C}^{m\times m} $ be a skew-symmetric matrix, i.e. $M^T = -M$.  
In the even case, i.e. $m= 2 n $ for some $ n \in \mathbb{N} $, the pfaffian is defined by 
\begin{equation} \label{def:pfaf}
{\rm pf}[M] = \frac{1}{2^n n!} \sum_{\pi \in S_{2n}} {\rm sgn}(\pi) \prod_{j=1}^n a_{\pi(2j-1), \pi(2j)}
\end{equation}
where $S_{2n}$ is the symmetric group of permutations and ${\rm sgn}( \pi)$ is the sign of
the permutation $\pi \in S_{2n}$. (Taking the skew-symmetry into account this definition is seen to coincide with~\eqref{def:quasifree}.) 
The pfaffian of any skew-symmetric  matrix with $m$ odd is defined to be $0$. 
It is also convention to define the pfaffian of a $0 \times 0$ matrix to be $1$. 

Pfaffians share many similarities with determinants. First, they are invariant under certain elementary row operations
which must be partnered with corresponding column operations to preserve skew-symmetry:
 \begin{enumerate}
 \item Let $\tilde{M}$ be the matrix obtained from $M$ by multiplying a given row and the corresponding column of 
 $M$ by a constant $\lambda$. Then ${\rm pf}[\tilde{M}] = \lambda {\rm pf}[M]$. 
 \item Let $\tilde{M}$ be the matrix obtained from $M$ by simultaneously interchanging two distinct 
 rows and the corresponding columns. Then ${\rm pf}[\tilde{M}] = - {\rm pf}[M]$.  
 \item Let $\tilde{M}$ be the matrix obtained from $M$ by taking a multiple of a given row and the corresponding column and
 adding it to another row and the corresponding column. Then ${\rm pf}[\tilde{M}] = {\rm pf}[M]$.
\end{enumerate}

Next, pfaffians satisfy a Laplace expansion. The simplest case is an expansion along the first row/column,
\begin{equation} \label{eq:Laplace}
{\rm pf}[M] = \sum_{\ell =2}^{2m} m_{1, \ell} (-1)^{\ell} {\rm pf}[M_{\hat{1} \hat{\ell}}] \, , 
\end{equation}
where $M_{\hat{1} \hat{\ell}}$ is the sub-matrix obtained from $M$ by simultaneously removing two rows and two columns;
namely those corresponding to $1$ and $\ell$.\\

Our new estimate concerns pfaffians of skew-symmetric matrices $M \in \mathbb{C}^{2(n+1) \times 2(n+1)}$ 
with the following block structure:
\begin{equation}\label{eq:skewPfM}
M = \left( \begin{matrix}  0 & \alpha & v_1^T & v_2^T \\
 &0 & w_1^T & w_2^T \\
& & A & B \\
&  &  & C \end{matrix} \right) 
\end{equation}
where $ \alpha \in \mathbb{C} $, the columns $ v_1, w_1 \in  \mathbb{C}^{2p} $ while $ v_2, w_2 \in  \mathbb{C}^{2q}$ with $p+q=n$,
and the blocks  $ A \in  \mathbb{C}^{2p\times2p} $, $ B    \in  \mathbb{C}^{2p\times2q} $, and $ C \in  \mathbb{C}^{2q\times2q}$ with
both $A$ and $C$ also skew-symmetric. As the remainder of the matrix is determined through skew-symmetry, we leave it blank.

We will assume that $ M $ models a correlation matrix which, in particular, entails that the modulus of its pfaffian is bounded. More generally, the following notion is tailored for our purposes.
\begin{definition}
A skew-symmetric matrix $ M \in \mathbb{C}^{2n\times2n} $ is said to have a correlation structure of depth $ k \in \{0,1,\dots , n\} $ with constant $ M_0 \in (0,\infty) $ if the pfaffian of all sub-matrices $ M_{\hat j_1, \hat j_2 , \dots, \hat j_{2l} }$ which result from simultaneously eliminating the rows and columns labeled $ j_1,  j_2 , \dots,  j_{2l} $  satisfy
\begin{equation}
\left| \pf M_{\hat j_1, \hat j_2 , \dots, \hat j_{2l}}  \right| \leq M_0
\end{equation}
for all disjoint integers $  j_1,  j_2 , \dots, j_{2l} \in \{ 1, \dots , 2n \} $ and all 
$ l \in \{ 0, \dots , k\}  $. The case $ l = 0 $ by definition corresponds to no eliminations, i.e. the bound $|\pf M| \leq M_0 $.  
\end{definition} 
We then have the following  result. 
\begin{theorem}\label{thm:Pfaffians}
Let $ M \in   \mathbb{C}^{2(n+1)\times2(n+1)} $ be a skew-symmetric matrix which has the block-structure~\eqref{eq:skewPfM} and a correlation structure of depth $ 2 $ with constant $ M_0$.
Then
\begin{equation} \label{thm:pfbd}
\left| \pf M \right| \leq M_0  \left( |\alpha| + \| v_2 \|_1+  \| v_1 \|_1 \| w_1 \|_1  + \| v_1 \|_1 \| w_2 \|_1 \sum_{j=1}^{2p} \| r_j(B) \|_1 \right) \, , 
\end{equation} 
where $ r_j(B) \in \mathbb{C}^{2q} $ are the row vectors of $  B    \in  \mathbb{C}^{2p\times2q} $ and $ \| \cdot \|_1 $ denotes the $ 1 $-norm.  
\end{theorem}
The proof is based on two lemmas. The first is a straightforward implication of the multi-linearity of pfaffians as expressed in~\eqref{eq:Laplace}.
\begin{lemma}\label{lem:LaplacePf}
In the situation of Theorem~\ref{thm:Pfaffians}:
\begin{equation}\label{eq:LaplacePf}
\left| \pf M \right| \leq  ( |\alpha| + \| v_2 \|_1+ \| v_1 \|_1 \| w_1 \|_1) M_0  + \| v_1 \|_1 \| w_2 \|_1 \sup_{\substack{j \in \{1, \dots, 2p\} \\ k \in  \{2p+1, \dots, 2(p+q)\}  }} \left| \pf \left( \begin{matrix} A & B \\ & C \end{matrix} \right)_{\hat j , \hat k} \right| \, .  
\end{equation}
\end{lemma}
\begin{proof}
An application of the Laplace expansion~\eqref{eq:Laplace} yields
\begin{equation}\label{eq:Lemma1Lapex}
\pf M = \alpha \pf M_{\hat 1 , \hat 2 } + \sum_{j = 1}^{2p}  (-1)^j (v_1)_j  \pf M_{  \hat 1 , \widehat{ (2+j)}}  + \sum_{j = 1}^{2q}  (-1)^j (v_2)_j  \pf M_{  \hat 1 , \widehat {(2(p+1)+j)}} \, . 
\end{equation}
Using the assumed correlation structure, the first and third term above are bounded by $| \alpha| M_0 $ and $\| v_2 \|_1 M_0 $, respectively.
For the remaining sum, we again Laplace expand along the first row of $ M_{  \hat 1 , \widehat{ (2+j)}} $:
\begin{align}\label{eq:Lemma1Lapex2}
\pf M_{  \hat 1 , \widehat{ (2+j)}} & = \pf \left( \begin{matrix} 0 & w_1^T & w_2^T \\ & A & B \\ & & C \end{matrix} \right)_{\widehat{ (1+j)} }  \\
& = \sum_{\substack{ k= 1 \\ k \neq j }}^{2p} (-1)^{\epsilon_j(k)} (w_1)_k   \pf M_{  \hat 1 , \hat 2 , \widehat{ (2+j)}, \widehat {(2+k) }}  + \sum_{k=1}^{2q} (-1)^k  (w_2)_k \pf \left( \begin{matrix}  A & B \\  & C \end{matrix} \right)_{\hat j , \widehat (2p+k) } \, ,\notag 
\end{align}
with suitable exponents $ \epsilon_j(k) \in \{ 0,1\} $. The first sum is bounded by $ \|w_1\|_1 M_0 $.  This gives rise to the third term in~\eqref{eq:LaplacePf}.  
The second sum in \label{eq:Lemma1Lapex} is 
bounded by $ \| w_2 \|_1 $ times the supremum in the right side of~\eqref{eq:LaplacePf}.
\end{proof}

It remains to estimate the pfaffians which appear on the right side of \eqref{eq:LaplacePf}.
\begin{lemma}\label{eq:OffDPf}
Consider a skew-symmetric matrix $ D \in \mathbb{C}^{2n \times 2n} $ of the form
\begin{equation}
D = \begin{pmatrix} A & B \\ \mbox{} & C \end{pmatrix}  \, , 
\end{equation} 
with blocks $A \in \mathbb{C}^{(2p-1) \times (2p-1)} $, $ B \in \mathbb{C}^{(2p-1) \times (2q-1)} $, and $C \in \mathbb{C}^{(2q-1) \times (2q-1)}$. 
(Here all integers $n,p,q \geq 1$ and $n = p+q-1$.) If, in addition, $D$ has a correlation structure of depth $1$ with constant $ K $, then 
\begin{equation}
\left| \pf D  \right| \leq K  \sum_{j=1}^{2p-1} \| r_j(B) \|_1 \, 
\end{equation}
where $r_j(B)$ denotes the $j$-th row of $B$.
\end{lemma}
\begin{proof}
Since $A \in \mathbb{C}^{(2p-1)\times(2p-1)}$ is skew-symmetric, its kernel is non-trivial. 
By the rank-nullity theorem, the range of $A$ has dimension smaller or equal to $ 2p-2$,
and consequently the columns of $A$ and the rows of $A$
are linearly dependent. Thus, there are $ \mu_1, \dots , \mu_{2p-1} \in \mathbb{C}$, not all zero, for which
\begin{equation} \label{ldrows}
\sum_{j=1}^{2p-1} \mu_j r_j(A) = 0 \,  
\end{equation}
where we have denoted by $r_j(A)$ the $j$-th row of $A$. As the $\mu_j$ do not all vanish, choose 
$ j_0 \in \{ 1, \dots , 2p-1 \} $ satisfying $ |\mu_j| \leq |\mu_{j_0} | \neq 0 $ for all $ j $. 
Without loss of generality, we will assume that $ j_0 = 1 $. Since the pfaffian is invariant under joint row/column operations, 
we may use (\ref{ldrows}) to eliminate the first row/column of $D$. In fact,
\begin{equation}\label{eq:Lem2row}
\pf D = \pf \left( \begin{matrix}  0 & 0^T & b_1^T \\ & \widehat A & \widehat B \\ & & C \end{matrix}\right)
\end{equation}
where $ \widehat A = A_{\hat 1} \in \mathbb{C}^{(2p-2)\times (2p-2) } $ is the sub-matrix of $ A $ obtained by deleting the first row and column and 
$  \widehat B \in \mathbb{C}^{(2p-2)\times (2q-1) } $ is the matrix with row vectors $ r_2(B) , \dots , r_{2p-1}(B) \in \mathbb{C}^{2q-1} $. 
Moreover, by inspection, it is clear that the first row is given by
\begin{equation}
b_1^T = r_1(B) + \sum_{j \geq 2} \frac{\mu_j}{\mu_1} r_j(B) \, . 
\end{equation}
We now Laplace expand along the first row on the right side of \eqref{eq:Lem2row} and obtain
\begin{equation}
\pf D  = \sum_{k=1}^{2q-1} (-1)^{k} (b_1 )_k \pf A_{\hat 1, \widehat{2p-1+k}} \, .  
\end{equation}
Using the correlation structure for $D$ and the fact that $ |\mu_j|\leq |\mu_1| $ we arrive at
\begin{equation}
 | \pf D | \leq K \, \| b_1 \|_1 \leq K  \sum_{j=1}^{2p-1} \| r_j(B) \|_1 \, , 
 \end{equation}
 which concludes the proof.
\end{proof}

Combining these two lemmas, there is a short proof of Theorem~\ref{thm:Pfaffians}.

\begin{proof}[Proof of Theorem~\ref{thm:Pfaffians}:]
Using Lemma~\ref{lem:LaplacePf}, it is clear that we need only estimate the supremum on the right side of~\eqref{eq:LaplacePf}. 
Each of the corresponding pfaffians is of a skew-symmetric matrix satisfying the assumptions of Lemma~\ref{eq:OffDPf} with $ K = M_0 $.
\end{proof}

\section{Proof of decay of multipoint correlation functionals}

\subsection{Proof of Theorems~\ref{thm:det} and \ref{thm:detrandom}}
We organize the proofs of Theorems~\ref{thm:det} and \ref{thm:detrandom} similarly. 
In each, for any $n \geq 1$, we consider a pair of configurations 
$x = (x_1, x_2, \cdots, x_n)$ and $y = (y_1, y_2, \cdots, y_n)$ both in $\mathbb{Z}^n$ which we assume to be fermionic and ordered, cf.~\eqref{eq:order1}. 
The configuation's distance is attained at an optimizing pair $ j_o \in \{ 1, \dots , n \} $, i.e.
\begin{equation} \label{simdist}
D(x , y ) =  \max_{1 \leq j \leq n} |x_j - y_j| = | x_{j_o} - y_{j_o} | \, . 
\end{equation}
Without loss of generality, we will assume that
\begin{equation} \label{xleft}
x_{j_o} \leq  y_{j_o} 
\end{equation}
since the roles of $ x $ and $ y $ may be interchanged in the case that $ x_{j_o} >  y_{j_o} $.  For convenience of notation we will relabel the particles 
\begin{equation} \label{relabel}
 j_o \mapsto 1 \, , \quad  (1, \dots , j_o-1) \mapsto (2, \dots , j_o) \, , \quad\mbox{and}\quad   (j_o+1,\dots, n ) \mapsto (j_o+1,\dots, n ) \, . 
 \end{equation}
After this relabeling the $ n \times n $ correlation matrix which we are interested in is given by
$
M = \left( \omega(j,k) \right)_{1\leq j,k\leq n}$ with $ \omega(j, k) :=  \langle \delta_{x_j}, \rho(s_j,t_k)  \delta_{y_k} \rangle $.
It has the following structure:
\begin{equation}\label{eq:defM}
M = \ \left( \begin{matrix}  \alpha & v_1^T & v_2^T \\
w_1 & A & B \\
w_2 & C & D \end{matrix} \right) 
\end{equation}
where we have set 
\begin{align} \label{label}
&  \alpha :=  \omega(1 ,1 )  \, , \quad  v  =   \left(  \begin{matrix} \omega(1 ,2)  \\ \vdots \\   \omega(1 ,n) \end{matrix} \right) \, , \quad
  w =   \left(  \begin{matrix}   \omega(2 , 1 )  \\ \vdots \\      \omega(n , 1 )  \end{matrix} \right) \, , \notag \\
 \mbox{and} & \quad  \left( \begin{matrix} A & B \\  C & D \end{matrix} \right)  \  =  \ \left(  \omega(j, k)\right)_{2\leq j , k \leq n } \, . 
\end{align}
 The sub-decomposition of the vectors $ v = (v_1,v_2)^T $ and $ w = (w_1,w_2)^T $ and thereby also the matrix $M$ into the
 blocks $A, B, C, D$ is done according to the following rule in which we use the above relabeling
\begin{align}\label{eq:labbelingrule}
& \mbox{$\omega(1, k)$ is a component of $v_1 $ if and only if } y_k \leq x_{1} + \tfrac{1}{2} \, d(x_{1} , y_{1} ) \, , \notag \\
&  \mbox{$\omega(k, 1)$ is a component of $w_1 $ if and only if } y_k \leq x_{1} + \tfrac{1}{2} \, d(x_{1} , y_{1} ) \, .
\end{align}
This renders $ A $ and $ D $ square matrices. Moreover, we have the following estimates: 
\begin{lemma}\label{lem:vec}
Let $x$ and $y$ be fermonic, ordered configurations in $\mathbb{Z}^n$, i.e. such that (\ref{eq:order1}) holds, and
set $ \delta := D(x,y)$. With respect to the relabeling introduced in (\ref{simdist}) - (\ref{eq:labbelingrule}) above, one has
\begin{align}
\| v_2 \| \ & \leq \sum_{\ell > \delta/ 2}  \hat \rho(x_1,x_1+\ell)   \, , \quad
\| w_1 \|  \leq \sum_{\ell > \delta}   \hat \rho(y_1-\ell ,y_1)     \label{eq:w11}  \\
\| B \| &   \leq  \sum_{\ell \geq 1} \sum_{\ell' > \delta/2}    \hat \rho(x_{1}-\ell,x_{1}+\ell') \, , \label{eq:nB}
\end{align}
where $ \hat \rho(x_j,y_k) := \sup_{s,t\in I } | \langle \delta_{x_j}, \rho(s,t)  \delta_{y_k} \rangle | $ and $\ell, \ell' \in \mathbb{Z}$.
\end{lemma}
\begin{proof}
Recall that wlog $\delta = y_1 - x_1$, see (\ref{xleft}). By the labeling rule~\eqref{eq:labbelingrule}, we have that:
\begin{enumerate}
\item The components of $v_2 $ correspond to $ y_k > x_{1} + \delta/2 $. Thus,
\begin{equation}
\| v_2 \|  = \sqrt{ \sum_{\stackrel{k :}{ y_k > x_1+ \delta/2}} | \langle \delta_{x_1}, \rho(s_1,t_k)  \delta_{y_k} \rangle |^2 } \leq \sum_{\stackrel{k :}{ y_k > x_1 + \delta/2}}\hat \rho(x_1,y_k)  \, . 
\end{equation}
Since the configuration $ y $ is assumed to be fermionic, i.e., $ y_j \neq y_k $ for all $ j \neq k $, the right side above is trivially estimated by the first term in~\eqref{eq:w11}. 
\item  The components of $w_1$ correspond to $ y_j \leq x_1 +  \delta/2 < y_1 $. In this case, it must be that $ x_j < x_{1} $. As a result,
\begin{equation}
\| w_1 \|  \leq \sum_{\stackrel{j :}{ x_j < x_{1}}}  \hat \rho(x_j,y_{1}) 
\end{equation}
similar to before. This is clearly bounded by the right side of~\eqref{eq:w11}.
\item The components of $B$ are of the form $\omega(j, k)$ corresponding to 
$ y_k > x_1 + \delta/2 $ and $ x_j < x_1 $. Since the operator norm of $ B $ is trivially bounded by 
the Frobenius norm $ \| B \|_2 $ and moreover,
\begin{equation}
 \| B \|_2 = \sqrt{ \sum_{\stackrel{j :}{ x_j < x_{1}}}  \sum_{\stackrel{k :}{ y_k > x_{1} + \delta/2}}  | \langle \delta_{x_j}, \rho(s_j,t_k)  \delta_{y_k} \rangle  |^2 } \leq \sum_{\stackrel{j :}{ x_j < x_{1}}}  \sum_{\stackrel{k :}{ y_k > x_{1} + \delta/2}} \hat \rho( x_j,y_k)  \, ,
\end{equation}
(\ref{eq:nB}) readily follows. 
\end{enumerate} 
\end{proof}

Using the above lemma and our general estimate for determinants (Theorem~\ref{Had2.0}), we can now easily prove
Theorems~\ref{thm:det} and \ref{thm:detrandom}.

\begin{proof}[Proof of Theorem~\ref{thm:det}]
Fix a pair of fermonic, ordered configurations $x$ and $y$ in $\mathbb{Z}^n$. With respect to the relabeling
described in (\ref{simdist})- (\ref{eq:labbelingrule}) above, denote by $M$ the corresponding correlation matrix which has the
structure of \eqref{eq:defM}. Since $ \| \rho(s,t) \| \leq 1 $, it is clear that $ \| M \| \leq 1 $. 
Using Theorem~\ref{Had2.0}, we have the estimate
\begin{equation}
|\det M | \leq  \  | \alpha| + \| v_2 \| + \| w_1 \| + \|B \| + 2 \sqrt{ \| v_1 \|( \| w_1\| + \| B \|)}  \, .
\end{equation}
With Lemma~\ref{lem:vec} and the assumed decay, i.e. \eqref{eq:expdecay}, we have with $ \delta := D(x,y)$, 
\begin{align}\label{eq:detlastest}
| \alpha| + \| v_2 \| + \| w_1 \| + \|B \| & \leq C e^{- \mu K(\delta) } + C \sum_{\ell =0}^{\infty} e^{- \mu K(l +\delta/2)} \notag \\
& \quad + C \sum_{\ell =0}^{\infty} e^{- \mu K(l +\delta)} + C \sum_{\ell, \ell' =0}^{\infty} e^{- \mu K ( \ell + \ell'+  \delta/2)} \notag \\
&\leq C e^{-(\mu-\mu_0) K(\delta/2) } \left( 1 + 2 \sum_{l=0}^\infty e^{- \mu_0 K(l)} +   \sum_{l=0}^\infty (1+\ell) e^{- \mu_0 K(\ell)} \right) \notag \\
& \leq 4 C I(\mu_0) \, e^{-(\mu-\mu_0) K(\delta/2) } \, . 
\end{align}
Similarly, since $\| v_1 \| \leq \| r_1(M) \| \leq \|M \| \leq 1$, the bound
\begin{eqnarray}
2 \sqrt{ \| v_1 \|( \| w_1\| + \| B \|)} & \leq &  2 \sqrt{ \| w_1 \| + \| B \|} \nonumber \\
& \leq & 4 \sqrt{C I(\mu_0)} \, e^{- (\mu -\mu_0) K(\delta/2)/2} \, . 
\end{eqnarray}
follows. This completes the proof.
\end{proof}

Proceeding similarly, one has the following.
\begin{proof}[Proof of Theorem~\ref{thm:detrandom}]
We again denote by $M =M_{s,t} $ the correlation matrix (making the time-dependence explicit).
Arguing as above, the bound
\begin{align}
&\mathbb{E} \left( \sup_{s,t \in I^n} |\det M_{s,t}| \right)  \leq  \mathbb{E} \left( \sup_{s, t \in I} | \alpha_{s,t} | \right) +  \mathbb{E} \left( \sup_{s\in I , t\in I^{n} } \| (v_{s,t})_2\| \right) +  \mathbb{E} \left( \sup_{t\in I , s\in I^{n} } \| (w_{s,t})_1 \| \right) \notag \\
&   +  \mathbb{E} \left( \sup_{s,t \in I^n} \| B_{s,t} \| \right) + 2 \sqrt{ \mathbb{E} \left( \sup_{t \in I, s\in I^n } \| (w_{s,t})_1 \| + \sup_{s,t \in I^n} \| B_{s,t} \| \right)} 
\end{align}
readily follows from an application of Jensen's inequality. With Lemma~\ref{lem:vec} and the
a priori estimate \eqref{ass:dis}, we find that
\begin{eqnarray}
\mathbb{E} \left( \sup_{s,t \in I^n} | {\rm det}M_{s,t}| \right) & \leq & C e^{ - \mu \delta} + C \frac{ e^{ - \mu \frac{\delta}{2} } }{1-e^{-\mu}} + C \frac{ e^{ - \mu \delta} }{1-e^{-\mu}} \nonumber \\
& \mbox{ } & \quad +  C \frac{ e^{ - \mu \frac{\delta}{2} } }{(1-e^{-\mu})^2} + 2\sqrt{ C \frac{ e^{ - \mu \delta } }{1-e^{-\mu}} + C \frac{ e^{ - \mu \frac{\delta}{2} } }{(1-e^{-\mu})^2}}
\nonumber \\
& \leq & \frac{8 \max\{ C, \sqrt{C} \}}{(1-e^{- \mu})^2} e^{- \mu \frac{\delta}{4}}
\end{eqnarray}
This completes the proof.
\end{proof}

\subsection{Proof of Theorem~\ref{thm:pf} and \ref{thm:majodis}}

Without loss of generality we will assume that the configuration $ x \in \mathbb{Z}^{2n} $ is ordered according to~\eqref{eq:ordermaj}. 
There is some $ j_o \in \{ 1, \dots , n \}  $ such that $ r(x) = | x_{2j_o} - x_{2j_o-1} | $. 
We may now relabel $ j_o \mapsto 1 $, $ (1, \dots , j_o-1) \mapsto (2, \dots , j_o) $, and $ (j_o+1,\dots, n ) \mapsto (j_o+1,\dots, n ) $, such that the skew-symmetric $ 2n \times 2n $ matrix featuring in Theorem~\ref{thm:pf}, which has entries $  \omega(j, k) :=  \omega\left(   a_{x_j}^{\#_j}(t_j)  a_{x_k}^{\#_k}(t_k) \right)$, 
can be assumed to have the following block structure
\begin{equation}\label{eq:defM1}
M \ = \  \left( \begin{matrix}  0 & \alpha & v_1^T & v_2^T \\
 &0 & w_1^T & w_2^T \\
& & A & B \\
&  &  & C \end{matrix} \right) 
\end{equation}
with $  \alpha := \omega(1,2)  $ and
\begin{align}
 &  v_1  =   \left(  \begin{matrix} \omega(1,3)   \\ \vdots \\   \omega(1,2j_o) \end{matrix} \right) \, , \; v_2  =   \left(  \begin{matrix} \omega(1,2j_o+1)   \\ \vdots \\   \omega(1,2n) \end{matrix} \right) \, \;
   w_1 =   \left(  \begin{matrix}   \omega(2 , 3 )  \\ \vdots \\     \omega(2 , 2j_o)  \end{matrix} \right) \, , \; w_2  =   \left(  \begin{matrix} \omega(2,2j_o+1)   \\ \vdots \\   \omega(2,2n) \end{matrix} \right) \notag \\
   &  \mkern50mu \mbox{and} \quad \left( \begin{matrix} A & B \\   & C \end{matrix} \right)  \  =  \ \left(  \omega(j,k) \right)_{3\leq j < k \leq 2n } \, . 
\end{align}
This decomposition has the property that the following norms are small ($ e^{-\mu K(r(x))} $) under the assumed decay of the two-point function:
\begin{align}
& \| v_2 \|_1 \  := \sum_{k = 2j_o+1}^{2n}  |  \omega(1,k) | \, , \; \| w_1 \|_1  := \sum_{k=3}^{2j_o} |   \omega(2,k) |    \label{eq:w1}  \\
& \| B \|_{\infty,\infty} :=  \sum_{k=3}^{2j_o} \sum_{l = 2j_o+1}^{2n} |  \omega(k,l) | \, . 
\end{align}
More precisely, we have
\begin{lemma}\label{lem:vec33}
For any Majorana configuration labelled such that \eqref{eq:ordermaj}  holds and $  r(x) = |x_{2j_o}-x_{2j_o-1}| $, we have
\begin{align}
&\| v_1 \|_1 \  \leq 2 \sum_{\ell=0}^\infty  \hat \rho(x_{2j_o-1},x_{2j_o-1}-\ell)  \, , \quad  && \| v_2 \|_1 \  \leq 2 \sum_{\ell =0}^\infty  \hat \rho(x_{2j_o-1},x_{2j_o}+\ell)   \, , \notag \\
&\| w_1 \|_1   \leq 2  \sum_{\ell =0}^\infty    \hat \rho(x_{2j_o} ,x_{2j_o-1}-\ell)     \, , \quad && \| w_2 \|_1 \  \leq 2 \sum_{\ell =0}^\infty  \hat \rho(x_{2j_o},x_{2j_o}+\ell)   \notag  \\
& \| B \|_{\infty,\infty}    \leq 4 \sum_{\ell , \ell=0}^\infty   \hat \rho(x_{2j_o-1}-\ell,x_{2j_o}+\ell') \, , && \label{eq:nB33}
\end{align}
where $ \hat \rho(x_j,y_k) := \max_{\#,\flat \in \{\pm\}} \sup_{s,t\in I } | \langle \delta_{x_j}^\#, \rho(s,t)  \delta_{y_k}^\flat \rangle | $.
\end{lemma}
\begin{proof}
Since there are no more than two Majorana Fermions (in fact of opposite flavor $ \# $) on each lattice site, the claim immediately follows from the definition of the vectors. 
\end{proof}

Using the estimate~\eqref{eq:expdecay} on $  \hat \rho $ we may hence conclude (in a similar fashion as in~\eqref{eq:detlastest}):
\begin{align}\label{eq:pfbvec331}
\max\{ \| v_1 \|_1 , \| w_2 \|_1 \}  & \leq 2 C \sum_{\ell=0}^\infty e^{-\mu K(\ell)} \leq 2 C I(\mu_0) \, , \notag \\ 
\max\{ \| v_2 \|_1 , \| w_1 \|_1 \} & \leq 2 C \sum_{\ell=0}^\infty e^{-\mu K(r(x) + \ell)} \leq 2 C I(\mu_0) \, e^{-(\mu-\mu_0) K(r(x)) } \, , \notag \\
 \| B \|_{\infty,\infty}  & \leq 4 C   \sum_{\ell,\ell' =0}^\infty e^{-\mu K(r(x) + \ell+\ell')} \leq 4 C  I(\mu_0) \, e^{-(\mu-\mu_0) K(r(x)) } \, .  
 \end{align}
These estimates allows to apply Theorem~\ref{thm:Pfaffians} to the pfaffian of~\eqref{eq:defM1} and hence give a proof of Theorem~\ref{thm:pf}.  
\begin{proof}[Proof of Theorem~\ref{thm:pf}]
As discussed above the skew-symmetric $ 2n \times 2n $ matrix featuring in Theorem~\ref{thm:pf} has the block structure~\eqref{eq:defM1},
\begin{equation}
\left| \pf\left( \omega\left(   a_{x_j}^{\#_j}(t_j)  a_{x_k}^{\#_k}(t_k) \right)\right)_{1\leq j < k \leq 2n }    \right|  = \left| \pf M \right| \, . 
\end{equation}
Since $ (\omega,\tau)  $ is assumed to be a quasi-free pair and $ \omega $ is bounded by $ M_0 $ as a functional, the matrix $ M $ has a  correlation structure of arbitrary depth with constant $ M_0 $, i.e.
\begin{equation}
\left| \pf M_{\hat j_1, \hat j_2 , \dots, \hat j_{2l}}  \right| = \left| \omega\left( a_{x_1}^{\#_1}(t_1) \dots a_{x_{j_1}}^{\#_{j_1}}(t_{j_1})^2\dots a_{x_{2l}}^{\#_{j_{2l}}}(t_{j_{2l}})^2 \dots a_{x_{2n}}^{\#_{2n}}(t_{2n})  \right) \right| \leq  M_0 \, .
\end{equation}
Here we used the fact that eliminating rows and colums in the Pfaffian is equivalent to inserting the corresponding Majorana operators since $(a_x^\#(t))^2 = 1 $. Moreover, the Majorana operators are bounded by one. We may hence apply Theorem~\ref{thm:Pfaffians}. 

The norms of the bordering vectors have been estimated in~\eqref{eq:pfbvec331}. Since the sums of the norms of the row-vectors of $ B $ are bounded according to
$ \sum_j \| r_j(B) \|_1 \leq  \| B \|_{\infty,\infty} $, we conclude
\begin{multline}
	\left| \pf M \right| \leq M_0 \left( C e^{-\mu K( r(x)) } + 2 C I(\mu_0) \, (1 + 2 C I(\mu_0) ) \,  e^{-(\mu-\mu_0) K(r(x)) }  \right. \\
	\left.
	+ 16 C^3 I(\mu_0)^3  \, e^{-(\mu-\mu_0) K(r(x)) }  \right)   \, . 
\end{multline}
This yields the claim.
\end{proof}

Proceeding similarly, one has the following.
\begin{proof}[Proof of Theorem~\ref{thm:majodis}] We again denote by $M =M_{s,t} $ the correlation matrix (making the time-dependence explicit).
Arguing as above, the bound
\begin{align}
&\mathbb{E} \left( \sup_{s,t \in I^n} \frac{|\pf M_{s,t}|}{M_0}  \right)  \leq \mathbb{E} \left( \sup_{s,t \in I^n} M_0^{-\frac{1}{3} }|\pf M_{s,t}|^\frac{1}{3} \right)  \notag \\
& \leq \mathbb{E} \left( \sup_{s, t \in I} | \alpha_{s,t} | \right)^{\frac{1}{3}} +  \mathbb{E} \left( \sup_{s\in I , t\in I^{n} } \| (v_{s,t})_2\|_1 \right)^{\frac{1}{3}}   \notag \\
& \quad  + \mathbb{E} \left( \sup_{s\in I , t\in I^{n} } \| (v_{s,t})_1\|_1 \right)^{\frac{1}{3}}  \mathbb{E} \left( \sup_{t\in I , s\in I^{n} } \| (w_{s,t})_1 \|_1 \right)^{\frac{1}{3}} \notag \\
& \quad  +   \mathbb{E} \left( \sup_{s\in I , t\in I^{n} } \| (v_{s,t})_1\|_1 \right)^{\frac{1}{3}}   \mathbb{E} \left( \sup_{s\in I , t\in I^{n} } \| (w_{s,t})_2\|_1 \right)^{\frac{1}{3}}   \mathbb{E} \left( \sup_{s,t \in I^n} \| B_{s,t} \|_{\infty,\infty} \right)^{\frac{1}{3}}  
\end{align}
follows from the bound $ |\pf M_{s,t}| \leq M_0 $ and H\"older's and/or Jensen's inequality. 
The expectation values in the right side may now be estimated using Lemma~\ref{lem:vec33} and the assumption~\eqref{eq:expdecay33}. Note that since the maximum in~\eqref{eq:expdecay33} is over a finite set, we may assume that this bound also holds for the maximum inside the expectation value.  
\end{proof}

{\bf Acknowledgements:} This work was partially supported by a grant from the Simons Foundation (\#301127 to Robert Sims) and from the DFG (WA 1699/2-1 to Simone Warzel). In addition, 
R. S. would like to acknowledge the hospitality of the Mathematics Department at the University of
California at Davis where much of this work was completed during his recent sabbatical.

\end{document}